\documentclass[11pt]{article}

\usepackage[centertags]{amsmath}
\usepackage{amsfonts}
\usepackage{amssymb}
\usepackage{amsthm}
\usepackage{newlfont}
\usepackage{epsfig}
\usepackage{amscd}
\usepackage{mathrsfs}

\newcommand{\RR}{{\mathbb R}}

\newcommand{\CC}{{\mathbb C}}
\newcommand{\beq}{\begin{equation}}
\newcommand{\eeq}{\end{equation}}
\newcommand{\ba}{\begin{array}}
\newcommand{\ea}{\end{array}}
\newcommand{\bea}{\begin{eqnarray}}
\newcommand{\eea}{\end{eqnarray}}

\usepackage{graphicx}
\usepackage{color}
\newcommand{\corr}[1]{{\color{black}{#1}}}
\usepackage{transparent} 

\newtheorem{theorem}{Theorem}

\newtheorem{lemma}{Lemma}[section]

\newtheorem{definition}{Definition}[section]
\newtheorem{example}{Example}[section]

\numberwithin{equation}{section}
\allowdisplaybreaks

\begin{document}

\begin{center}
{\large   \bf The direct problem for the perturbed  Kadomtsev-Petviashvili II one line solitons} 
 
\vskip 15pt

{\large  Derchyi Wu}

\vskip 5pt

{ Institute of Mathematics, Academia Sinica, 
Taipei, Taiwan}

e-mail: {\tt mawudc@gate.sinica.edu.tw}

\vskip 10pt

{\today}

\end{center}

\begin{abstract}
We provide rigorous analysis for the direct scattering theory   of perturbed  Kadomtsev-Petviashvili II one line solitons. Namely,  for generic small initial data, the existence of the eigenfunction is proved by establishing uniform estimates of the Green function and the Cauchy integral equation for the eigenfunction is justified by analysing the spectral transform.
\end{abstract}

\section{Introduction}\label{S:motivation}
\noindent 

The well-posedness problem of the Kadomtsev-Petviashvili II (KPII) equation 
\begin{equation}\label{E:KPII-intro}
\begin{gathered}
(4u_t+u_{xxx}-6uu_x)_x+3u_{yy}=0, \\
u(x,y,0)=u_N(x)+v_0(x,y),   
\end{gathered}
\end{equation}
where $u_N$ is an $N$ line soliton, has been initiated by Bourgain \cite{Bo93} and solved by Molinet-Saut-Tzvetkov \cite{MST11}. Their results show that the deviation of the KPII solution from the initial (perturbed) line soliton could grow exponentially  unbounded during the evolution which is not consistent with the  {isospectral property} of integrable systems. Excellent $L^2$-{  orbital stability  and $L^2$-                                                    instability  theories were established  by Mizumachi \cite{M15} for perturbed KPII one line solitons. But the approach has difficulties to be generalized to multi line solitons.   

The {\bf inverse scattering theory (IST)} is a powerful method to identify and solve classes of integrable nonlinear PDEs and integrable dynamical systems. Many important nonlinear PDEs, such as KdV, NLS, sine-Gordon, etc, have been studied by this mothod. Integrability of the KPII equation has been known since the beginning of the 1970s. It can be integrated via the {\bf Lax pair}
\begin{equation}\label{E:KPII-lax-1}
\begin{split}
\left\{ 
{\begin{array}{l}
\left(\partial_y-\partial_x^2+u\right)\Psi(x,y,t,\lambda)=0,\\
\left(\partial_t-\left(-\partial_x^3+\frac 32u\partial_x+\frac 34u_x+\frac 34\partial_x^{-1}u_y+(-i\lambda)^3\right)\right)\Psi(x,y,t,\lambda)=0.
\end{array}}
\right.
\end{split}
\end{equation}  The IST of 
 the spectral operator   $
\partial_y-\partial_x^2+u $   was solved, schematically, solvability of the spectral equation is transformed to that of a Cauchy integral equation  or $\overline\partial$-equation, by \cite{W87}, \cite{GN88}, \cite{G97} in case the potential  {$u$  rapidly decays at spatial infinity}. In particular, an $L^2_k$-stability theorem  can be implied by \cite{W87}. As  $u$  is a perturbed  multi line soliton, 
 two research groups  \cite{BPP97}-\cite{BP214}, \cite{VA04}   have   published substantial and important works  on algebraic characterization and formal IST. In particular, the most remarkable characteristic, discontinuities for the Green function and eigenfunction   had been discovered by Boiti, Pempenelli, Pogrebkov, and Prinari (cf \cite{BPPP01-p}, \cite{BP302}, \cite{BP214}). However, with incomplete rigorous analysis for the Green function and spectral data, both groups formulated  Cauchy integral equations which  hold only for non generic initial data \cite{BPPP01-p}, \cite{BP302}, \cite
 {BP214}, \cite
 {VA04}.  
 So the IST for perturbed KPII line solitons is still an important open problem in this field  \cite{P17}, \cite{S17}.  

This report is aimed to a rigorous analysis for the direct scattering theory   of perturbed  KPII one line solitons. Precisely, 
 a uniform estimate of the Green function is established by decomposing the kernel into   Gaussian parts, oscillatory parts, rational functions, and   regular parts and analysing them separately via different techniques. The same proof also characterizes the discontinuities at singularities. Furthermore, due to discontinuities of the eigenfunction, both discrete and continuous scattering data are not meromorphic which make the kernel of the Cauchy integral equation complicated to analyze. Inspired by \cite{BP214}, we provide a regularized eigenfunction to simplify the formula.  
 Still the scattering data  blow up at singularities which, along with  discontinuous, highly oscillating, and not fully symmetric kernels,  cause  difficulties for deriving uniform estimates for the spectral transform  and solving the inverse problem.  In this paper, we  only  present non-uniform estimates of the  spectral transform which is sufficient to imply a Cauchy integral equation provided the initial data satisfying
 \[
 \ba{c}
 (1+|x|+|y|)^2\partial_y^j\partial_x^kv _0\in  {L^1\cap L^\infty},\ \  0\le j,\ k\le 4,\ \ |v_0 |_{{L^1\cap L^\infty}}\ll 1.
 \ea
 \]

The contents of the paper are as follows. In Section \ref{S:IST-kp-eigenfunction}, for generic small initial data $v_0$, we derive a uniform estimate of the Green function and apply the result to prove the existence of the eigenfunction to KPII equation. In Section \ref{S:IST-kp-sd}, we extract the scattering data of eigenfunction $m$, introduce the regularized eigenfunction $\mathfrak m$, define    the scattering operator $T$, derive spectral analysis,  and justify a singular Cauchy integral equation.

{\bf Acknowledgments}. We feel very grateful to  Y. Kodama   for introducing  the  stability problem of KPII line solitons and T. Mizumachi and B. Prinari for their visits to   Taiwan and inputs on the KPII equation.  We would like to pay  respects to the pioneer works done by Boiti, Pempinelli, Pogrebkov, Prinari, Villarroel, and Ablowitz. Special acknowledgement needs to be made to  A. Pogrebkov and P. Grinevich due to numerous inspiring  discussion and valuable suggestion. This research project was  partially supported by NSC  105-2115-M-001-001- and 106-2115-M-001-002-.

\section{The forward problem I: a class of eigenfunctions}\label{S:IST-kp-eigenfunction}

We start the investigation of the inverse scattering theory of the KPII perturbed  one-line solitons. Consider the {\it{\bf spectral equation}}
\begin{equation}\label{E:kp-line-normal}
\ba{c}
(\partial_y-\partial_x^2+u(x,y))\Psi(x,y,\lambda)=0,\\
u(x,y)=u_0(x)+v_0(x,y),\\
u_0(x )=-2\corr{\kappa^2}\textrm{sech}^2\corr{\kappa}x,\, \kappa>0,\ 
v_0(x,y)\in\RR, 
\ea
\end{equation}
with the boundary condition
\begin{equation}\label{E:kp-line-normal-bdry-1}
\ba{c}
\lim_{(x,y)\to\infty}  (\Psi(x,y,\lambda)-\vartheta_-(x,\lambda)e^{(-i\lambda) x+(-i\lambda)^2y} )=0,
\ea
\end{equation}
where $\vartheta_-(x,\lambda)$ is one of eigenfunctions of the Schr$\ddot{\mbox{o}}$dinger operator corresponding to the KdV one soliton $u_0(x)$, i.e., \cite{MS91}
\beq\label{E:eigen-x}
\ba{rl}
\left(-\partial_x^2+u_0(x)-\lambda^2\right)f=0,& f=\phi_\pm(x,\lambda),\,\psi_\pm (x,\lambda),
\\
\phi_\pm(x, \lambda)=\varphi_\pm(x,\lambda) e^{\mp i\lambda x}, & \psi_\pm(x, \lambda)=\vartheta_\pm(x, \lambda) e^{\pm i\lambda x}, \\ \varphi_\pm=1\mp\frac {2i\corr{\kappa}}{\lambda\pm i\corr{\kappa}}\left(\frac { 1} {1 +e^{-2\corr{\kappa}x}} \right), & \vartheta_\pm= 1\mp \frac {  2i \corr{\kappa}} {\lambda\pm i\corr{\kappa}}\left(\frac {1} {1 +e^{2\corr{\kappa}x}}\right)  .
\ea
\end{equation}
 
Introducing the normalizations 
\begin{equation}\label{E:renormalization}
\ba{c}
\Psi(x,y,\lambda)
= \Phi(x,y,\lambda)e^{-\lambda^2 y}= m(x,y,\lambda)e^{(-i\lambda) x+(-i\lambda)^2y}
\ea
\end{equation} the spectral equation \eqref{E:kp-line-normal} turns into 
\beq\label{E:lax-normal}
\ba{c}
 L_{\lambda}\Phi=	\left(\partial_y-\partial_x^2+u_0(x)-\lambda^2\right)\Phi=-v_0(x,y)\Phi,\\
\mathcal L_\lambda m=\left(\partial_y-\partial_x^2+2i\lambda\partial_x+u_0(x)\right)m=-v_0(x,y)m.
\ea
\eeq	
Denote $G$ and $\widetilde G$ as the Green functions 
\beq\label{E:sym-0}
\ba{c}
 L_\lambda G(x,x',y-y',\lambda)=\delta(x-x')\delta(y-y'),\\
 \mathcal L_\lambda
 \widetilde G(x,x',y-y',\lambda)=\delta(x-x')\delta(y-y').
\ea
\eeq

Formula for the Green functions are available in \cite{BP302}, \cite[Eq.(3.1)]{BP214},  \cite[Eq.(17)]{VA04} for instance. For convenience, we sketch the approach in \cite{VA04} to derive $G(x,x',y,\lambda)$ in the following lemma.

\begin{lemma}\label{L:green-heat}    For $y\ne 0$, $\lambda\notin \mathbb R\cup i\mathbb R\cup \{\lambda\in\mathbb C|\lambda\pm i\corr{\kappa}\in\mathbb R\}$, 
\beq\label{E:green-1-d}
\ba{c}
G(x,x',y,\lambda)=G_c(x,x',y,\lambda)+G_d(x,x',y,\lambda),\\
G_{d}(x,x',y,\lambda)= -\ 2\theta(-y)\theta(\corr{\kappa}-|\lambda_I|) e^{(\lambda^2+\corr{\kappa^2})y\pm \corr{\kappa}(x-x')}\mathfrak g(x,x',\pm i\corr{\kappa}),\\
G_c=G_{\mathbb C^+}=G_{\mathbb C^-},\\
G_{\mathbb C^+}(x,x',y,\lambda)
=\int_{\mathbb R}\left(\theta(y)\chi_{-}-\theta(-y)\chi_{+}\right)e^{(\lambda^2-\left[\lambda+ \lambda'\right]^2)y}\\
 \times\frac{\phi_+(x,\lambda+\lambda')\psi_+(x',\lambda+\lambda')}{2\pi a(\lambda+\lambda')}d\lambda' \\
G_{\mathbb C^-}(x,x',y,\lambda)
=\int_{\mathbb R}\left(\theta(y)\chi_{-}-\theta(-y)\chi_{+}\right)e^{(\lambda^2-\left[\lambda+ \lambda'\right]^2)y}\\
 \times\frac{\phi_-(x',\lambda+\lambda')\psi_-(x,\lambda+\lambda')}{2\pi a(\lambda+\lambda')}d\lambda'. 
\end{array}
\end{equation}Here $\phi_\pm$, $\psi_\pm$ are defined by \eqref{E:eigen-x},
\beq\label{E:g-def}
\ba{c}
\mathfrak g(x,x',\lambda )=\left\{ 
{\begin{array}{l}\varphi_+ (x,\lambda )\vartheta_+ (x',\lambda ),\  \lambda\in \CC^+,\\
\varphi_- (x',\lambda )\vartheta_- (x,\lambda ),\  \lambda\in \CC^-,\end{array}}
\right.
\ea\eeq and 
\beq\label{E:chi-a}
\ba{c}
\textit{$\chi_{+}$  the characteristic function for $\{\lambda'|\,\textit{Re}(\lambda^2-\left[\lambda+\lambda'\right]^2)>0\}$,}\\
\textit{$\chi_{-}$  the characteristic function for  $\{\lambda'|\,\textit{Re}(\lambda^2-\left[\lambda+\lambda'\right]^2)<0\}$,}\\
a(\lambda)=\left\{ 
{\begin{array}{l}a_+ ( \lambda ),\  \lambda\in \CC^+,\\
a_- ( \lambda ),\  \lambda\in \CC^- ,\end{array}}
\right.  \ a_+(\lambda)=\frac {\lambda-i\corr{\kappa}}{\lambda+i\corr{\kappa}},\   a_-(\lambda)=\frac {\lambda+i\corr{\kappa}}{\lambda-i\corr{\kappa}},\\
\textit{$\theta(y)=1$  if  $y>0$ and vanishes elsewhere,}\\
\widetilde G(x,x',y,\lambda)=G(x,x',y,\lambda) e^{i\lambda(x-x')},\ \ 
\widetilde G=\widetilde G_c+\widetilde G_d.
\ea
\eeq 
\end{lemma}
\begin{proof} First of all, note that if the operator  $P=P(\frac{\partial^n}{\partial x^n})$ admits a complete set of eigenfunctions $\{\phi(x,\lambda)\}$,  i.e.,
\beq \nonumber
\ba{c}
P\phi(x,\lambda)=\lambda\phi(x,\lambda) ,\ \lambda\in\mathbb R, \\ 
\int_\RR\phi(x,\lambda)\phi(x',\lambda)d\lambda=\delta(x-x'),
\ea
\eeq
 then 
\beq\label{E:K}
\ba{c}
\left(\frac \partial{\partial y}+P\right)K(x,x',y-y')=\delta(x-x')\delta(y-y'), \\
K(x,x',y)=\int_\RR \left[\theta(y)\theta(\lambda)-\theta(-y)\theta(-\lambda)\right]e^{-y\lambda}\phi(x,\lambda)\phi(x',\lambda) d\lambda .
\ea
\eeq 

Secondly, from \eqref{E:eigen-x},
\beq\label{E:jost-adj}
\ba{c}
( -\partial_x^2+u_0-\lambda^2)\phi_\pm(x,\lambda+\lambda')=[(\lambda+\lambda')^2-\lambda^2]\phi_\pm(x,\lambda+\lambda'),\\
( -\partial_x^2+u_0-\lambda^2)\psi_\pm(x,\lambda+\lambda')=[(\lambda+\lambda')^2-\lambda^2]\psi_\pm(x,\lambda+\lambda'),\\
( -\partial_x^2+u_0-\lambda^2)\phi_\pm(x,\pm i\corr{\kappa})=-({\lambda }^2+\corr{\kappa^2})\phi_\pm(x,\pm i\corr{\kappa}),\\
( -\partial_x^2+u_0-\lambda^2)\psi_\pm(x,\pm i\corr{\kappa})=-({\lambda}^2+\corr{\kappa^2})\psi_\pm(x,\pm i\corr{\kappa}),
\ea
\eeq and by the residue theorem,  
\beq\label{E:c-+}
\begin{array}{r}
\frac 1{2\pi i}\int _{\mathbb R}\left[ \frac {\phi_+(x,\lambda+\lambda')\psi_+(x',\lambda+\lambda')}{a_+(\lambda+\lambda')} - e^{i(\lambda+\lambda')(x'-x)}\right]d\lambda'\\
=\textit{Res}_{\eta\in\mathbb C,\eta_I>\lambda_I>0}\left[\frac{\phi_+(x,\eta)\psi_+(x',\eta)}{a_+(\eta)}\right],\\
\frac 1{2\pi i}\int_{\mathbb R} \left[ \frac {\phi_-(x',\lambda+\lambda')\psi_-(x ,\lambda+\lambda')}{a_-(\lambda+\lambda')} - e^{i(\lambda+\lambda')(x'-x)}\right]d\lambda'\\=-\textit{Res}_{\eta\in\mathbb C,\eta_I<\lambda_I<0}\left[\frac{\phi_-(x',\eta)\psi_-(x,\eta)}{a_-(\eta)}\right].
\end{array}
\eeq  
Plugging
\[\begin{array}{l}
\frac 1{2\pi }\int_{\lambda+\lambda'\in\mathbb R}e^{i(\lambda+\lambda')(x'-x)}d\lambda'=\delta(x-x')
\end{array}
\]into \eqref{E:c-+}, 
we prove the following orthogonality and completeness identities  
\beq\label{E:jost-adj-new}
\begin{array}{ll}
& \delta({x-x'})\\
	=&	\left\{ 
{\begin{array}{l} \frac 1{2\pi}  \int_{\mathbb R} \frac {\phi_+(x,\lambda+\lambda')\psi_+(x',\lambda+\lambda')}{a_+(\lambda+\lambda')}d\lambda'+2\theta(\corr{\kappa}-\lambda_I) \phi_+(x,i\corr{\kappa})\psi_+(x',i\corr{\kappa}),\\
 	  \frac 1{2\pi}   \int_{\mathbb R} \frac {\phi_-(x',\lambda+\lambda')\psi_-(x,\lambda+\lambda')}{a_-(\lambda+\lambda')}d\lambda'+2 \theta(\corr{\kappa}+\lambda_I)\phi_-(x',-i\corr{\kappa})\psi_-(x,-i\corr{\kappa}) \end{array}}
\right.
	\end{array}
\eeq for $\lambda\in\CC^\pm$. Therefore the formula of $G(x,x',y,\lambda)$ follows from    \eqref{E:K}, \eqref{E:jost-adj}, and \eqref{E:jost-adj-new}.

The property $G_{\mathbb C^+}=G_{\mathbb C^-}$ follows from 
\beq\label{E:pm-sym-eigen}
\ba{c}
\frac{\phi_+}{a_+}(x, \lambda)=\psi_-(x, \lambda),\quad \frac{\phi_-}{a_-}(x, \lambda)=\psi_+(x, \lambda),
\ea
\eeq and 
\beq\label{E:a-1}
\ba{rl}
\textit{If $\lambda_R> 0$, then}
&\chi_{+}(\lambda')=\left\{ 
{\begin{array}{ll}
1,&-2\lambda_R<\lambda'< 0,\\
0, &\lambda'\ge 0,\,\lambda'\le -2\lambda_R;
\end{array}}
\right.\\
& \chi_{-}(\lambda')=\left\{ 
{\begin{array}{ll}
1,&\lambda'> 0,\,\lambda'< -2\lambda_R,\\
0,&-2\lambda_R\le\lambda'\le 0 ;
\end{array}}
\right.\\
\textit{If $\lambda_R< 0$, then}
&\chi_{+}(\lambda')=\left\{ 
{\begin{array}{ll}
1,&0<\lambda'< -2\lambda_R,\\
0, &\lambda'\le 0,\,\lambda'\ge -2\lambda_R;
\end{array}}
\right.\\
&\chi_{_-}(\lambda')=\left\{ 
{\begin{array}{ll}
1,&\lambda'< 0,\,\lambda'> -2\lambda_R,\\
0, &0\le\lambda'\le -2\lambda_R,
\end{array}}
\right.
\ea
\eeq which are independent of $\lambda\in\CC^+$ or $\CC^-$.
\end{proof}

 The basic existence theorem for this section is
\begin{theorem}\label{T:KP-eigen-existence} 
If $\partial_y^j\partial_x^kv_0\in L^1\cap L^\infty$, $0\le j,\,k\le 2$, $ {|v|_{L^1\cap L^\infty}}\ll 1$, then for fixed $\lambda\neq \pm i\corr{\kappa}$,   there is a unique solution $\Psi(x,y,\lambda)=m(x,y,\lambda)e^{(-i\lambda) x+(-i\lambda)^2y}$ to the problem \eqref{E:kp-line-normal}, \eqref{E:kp-line-normal-bdry-1} and $\partial_y^j\partial_x^km\in L^\infty$.
\end{theorem}

The proof of the theorem follows from the following uniform estimate
\begin{lemma}\label{L:eigen-green}
There exists a uniform constant $C$ such that the Green function $\widetilde G$, defined by \eqref{E:sym-0} - \eqref{E:chi-a}, satisfies
\beq
\ba{c}
|\widetilde G(x,x',y,\lambda)|\le C\left(1+\frac1{\sqrt{|y|}}\right)  \label{E:eigen-green}
\ea
\eeq for $\forall x,\,x',\,y\in\mathbb R$,  $y\ne 0$, $\lambda\notin \mathbb R\cup i\mathbb R\cup \{\lambda\in\mathbb C|\lambda\pm i\corr{\kappa}\in\mathbb R\}$.
\end{lemma}
\begin{proof} 

$\underline{\emph{Step 1}}:$ From Lemma \ref{L:green-heat} and \eqref{E:a-1}, 
\beq\label{E:discrete-est}
\ba{c}
|\widetilde G_d|= |- 2\theta(-y)\theta(\corr{\kappa}-|\lambda_I|) e^{(\lambda^2+\corr{\kappa^2})y\pm \corr{\kappa}(x-x')+i\lambda(x-x')}\mathfrak g(x,x',\pm i\corr{\kappa})|\\
= | 2\theta(-y)\theta(\corr{\kappa}-|\lambda_I|) \frac{e^{(\lambda^2_R+[\corr{\kappa}^2-\lambda_I^2])y {-\lambda_I(x-x')}}}{(e^{-\corr{\kappa}x}+e^{\corr{\kappa}x})(e^{-\corr{\kappa}x'}+e^{\corr{\kappa}x'})}|<C,
\ea
\eeq and
\beq\label{E:tilde-g-lambda}
\ba{rl}
&\widetilde G_c (x,x',y,\lambda)\\
=& \theta(\lambda_R)  {\theta(y) }   (\int_{-\infty}^{-2|\lambda_R|}+\int_{0}^\infty)\frac {e^{{-(\lambda'}^2+2\lambda\lambda') y+i\lambda'(x'-x)}\mathfrak g(x,x',\lambda+\lambda')}{2\pi a(\lambda+\lambda')}d\lambda'\\
&+ \theta(-\lambda_R)  {\theta(y) }   (\int_{-\infty}^{0}+\int_{2|\lambda_R|}^\infty)\frac {e^{{-(\lambda'}^2+2\lambda\lambda') y+i\lambda'(x'-x)}\mathfrak g(x,x',\lambda+\lambda')}{2\pi a(\lambda+\lambda')}d\lambda'\\
&- \theta(\lambda_R)  {\theta(-y) } \int_{-2|\lambda_R|}^{0}\frac {e^{{-(\lambda'}^2+2\lambda\lambda') y+i\lambda'(x'-x)}\mathfrak g(x,x',\lambda+\lambda')}{2\pi a(\lambda+\lambda')}d\lambda'\\
&- \theta(-\lambda_R)  {\theta(-y) } \int_{0}^{2|\lambda_R|}\frac {e^{{-(\lambda'}^2+2\lambda\lambda') y+i\lambda'(x'-x)}\mathfrak g(x,x',\lambda+\lambda')}{2\pi a(\lambda+\lambda')}d\lambda'  .
\ea
\eeq
By  a change of variables $\eta=\lambda_R+\lambda'$,
\beq\label{E:tilde-g-eta-new}
\ba{rl}
&\widetilde G_c (x,x',y,\lambda)\\
=& \frac {\theta(y)e^{-i\lambda_R(x'-x-2\lambda_Iy)}}{2\pi}\\
&\times (\int_{-\infty}^{-|\lambda_R|}+\int_{|\lambda_R|}^\infty)\frac {e^{(\lambda^2_R-\eta^2) y+i\eta(x'-x-2\lambda_I y)}\mathfrak g(x,x',\eta+i\lambda_I)}{a(\eta+i\lambda_I)}d\eta\\
&-  
 \frac{\theta(-y)e^{-i\lambda_R(x'-x-2\lambda_Iy)}}{2\pi} \\
 &\times\int_{-|\lambda_R|}^{|\lambda_R|}\frac{e^{(\lambda^2_R-\eta^2) y+i\eta(x'-x-2\lambda_I y)}\mathfrak g(x,x',\eta+i\lambda_I)}{a(\eta+i\lambda_I)}
d\eta\\
=&e^{-i\lambda_R(x'-x-2\lambda_Iy)}\int_{\mathbb R}\frac {\Xi(y,\lambda,\eta)\tilde{\mathfrak F}(x,x',y,\lambda,\eta)}{2\pi a(\eta+i\lambda_I)}d\eta,
\ea 
\eeq where
\beq\label{E:f}
\ba{rl}
&\tilde{\mathfrak F}(x,x',y,\lambda,\eta)= e^{(\lambda^2_R-\eta^2) y+i\eta(x'-x-2\lambda_I y)}\mathfrak g(x,x',\eta+i\lambda_I) ,\\
&\Xi(y,\lambda,\eta)=\theta(y)\chi_{\{-\infty,-|\lambda_R|\}\cup \{|\lambda_R|,\infty\}}-\theta(-y)\chi_{\{-|\lambda_R|,|\lambda_R|\}},
\end{array}
\eeq and $\chi_A$ denotes the characteristic function of the set $A$. Since singularities of  integrands are $
\lambda_I=\pm \kappa$, $ \lambda_R+\lambda'=0$, and   integrands consisting of both gaussian parts $ e^{(\lambda^2_R-\eta^2) y}$ and  oscillatory parts $ e^{ i\eta(x'-x-2\lambda_I y)}$,  we decompose
\beq\label{E:tilde-g-eta}
\ba{cc}
 \widetilde G_c (x,x',y,\lambda)=e^{ -i\lambda_R(x'-x-2\lambda_Iy)}(I+II+III+IV+V),&
 \textit{for $|\lambda_R|\le 1$},\\
 \widetilde G_c (x,x',y,\lambda)=e^{ -i\lambda_R(x'-x-2\lambda_Iy)}(I^\dagger+II^\dagger+III^\dagger),& 
 \textit{for $|\lambda_R|\ge 1$},
 \ea
\eeq with
\beq\label{E:green-decomposition}
\ba{rl}
I=&\pm\frac i \pi  \int_{-1}^1\Xi(y,\lambda,\eta) \frac { [e^{i\eta(x'-x-2\lambda_Iy)}-1]e^{(\lambda_R^2-\eta^2 ) y}\mathfrak g(x,x',\eta+i\lambda_I)}{ \eta+i(\lambda_I\mp \kappa)}d\eta,\\
II=&\pm\frac i \pi  \int_{-1}^1\Xi(y,\lambda,\eta)\frac { [e^{(\lambda_R^2-\eta^2 ) y}-1] \mathfrak g(x,x',\eta+i\lambda_I)}{ \eta+i(\lambda_I\mp \kappa)}d\eta,\\
III=& \pm\frac i \pi  \int_{-1}^1\Xi(y,\lambda,\eta)\frac {\mathfrak g(x,x',\eta+i\lambda_I)-\mathfrak g(x,x',\pm i\corr{\kappa}) }{ \eta+i(\lambda_I\mp \kappa)}d\eta,\\
IV=&\pm\frac i \pi \int_{-1}^1\Xi(y,\lambda,\eta)\frac { \mathfrak g(x,x',\pm i\corr{\kappa}) }{ \eta+i(\lambda_I\mp \kappa)}d\eta,\\
V=& (\int_{-\infty}^{-1}+\int_{1}^\infty)\frac {\Xi(y,\lambda,\eta)\tilde{\mathfrak F}(x,x',y,\lambda,\eta)}{2\pi a(\eta+i\lambda_I)}d\eta+  \int_{-1}^1\frac{\Xi(y,\lambda,\eta)\tilde{\mathfrak F}(x,x',y,\lambda,\eta)}{2\pi }d\eta,
\ea
\eeq and
\beq\label{E:tilde-g-eta-ge}
\ba{rl}
I^\dagger=& \theta(y) (\int_{-\infty}^{-|\lambda_R|}+\int_{|\lambda_R|}^\infty) \frac {  e^{i\eta(x'-x-2\lambda_Iy)+(\lambda_R^2-\eta^2 ) y}\mathfrak g(x,x',\eta+i\lambda_I)}{2\pi a(\eta+i\lambda_I)}d\eta,\\
  II^\dagger=& \mp\theta(-y) \int_{-1}^{1}  \frac {  e^{i\eta(x'-x-2\lambda_Iy)} e^{(\lambda_R^2-\eta^2 ) y}\mathfrak g(x,x',\eta+i\lambda_I)}{ \eta+i(\lambda_I\mp \kappa)}d\eta,\\
III^\dagger=& -\theta(-y) \{(\int_{-|\lambda_R|}^{-1} +\int_{1}^{|\lambda_R|})  \frac {  e^{i\eta(x'-x-2\lambda_Iy)+(\lambda_R^2-\eta^2 ) y}\mathfrak g(x,x',\eta+i\lambda_I)}{2\pi a(\eta+i\lambda_I)}d\eta\\
&+\int_{-1}^1\frac {  e^{i\eta(x'-x-2\lambda_Iy)+(\lambda_R^2-\eta^2 ) y}\mathfrak g(x,x',\eta+i\lambda_I)}{ 2\pi}\}. 
\ea
\eeq   

$\underline{\emph{Step 2  (Estimates for $III$ and $V$)}}:$ First of all, 
\beq\label{E:est-III}
\ba{c}
|III|\le C,
\ea
\eeq
since $\frac {\mathfrak g(x,x',\eta+i\lambda_I)-\mathfrak g(x,x',\pm i \corr{\kappa}) }{ \eta+i(\lambda_I\mp \kappa)}$ is a uniformly bounded function and $|\eta|\le 1$.  
  Moreover, note $|\frac 1{2\pi a(\eta+i\lambda_I)}|<C$ for $|\eta|\ge 1$. Therefore, via   \eqref{E:tilde-g-lambda}, for $\forall \lambda$,
\beq\label{E:gaussian}
\ba{rl}
& |V (x,x',y,\lambda)|\\
\le & C[\theta(\lambda_R)  {\theta(y) }   (\int_{-\infty}^{-2|\lambda_R|}+\int_{0}^\infty)e^{-({\lambda'}^2+2\lambda_R\lambda') y }d\lambda' \\
&+ \theta(-\lambda_R)  {\theta(y) }   (\int_{-\infty}^{0}+\int_{2|\lambda_R|}^\infty)e^{-({\lambda'}^2+2\lambda_R\lambda') y }d\lambda'\\
&+ \theta(\lambda_R)  {\theta(-y) } \int_{-2|\lambda_R|}^{0}e^{-({\lambda'}^2+2\lambda_R\lambda') y }d\lambda'\\
&+ \theta(-\lambda_R)  {\theta(-y) } \int^{2|\lambda_R|}_{0}e^{-({\lambda'}^2+2\lambda_R\lambda') y }d\lambda']\\
\le &  
C(\theta(y) \int_{0}^\infty e^{-y{s}^2}ds+\theta(-y) e^{y\lambda_R^2}\int_{0}^{|\lambda_R|} e^{-ys^2}ds)\\
\le & \frac C{\sqrt {|y|}}.
\ea
\eeq Here we have used 
\[
\ba{c}
  \theta(\lambda_R)  {\theta(y) } \int_{0}^\infty 
e^{-({\lambda'}^2+2\lambda_R\lambda') y }d\lambda'  
\le      {\theta(y) } \int_{0}^\infty 
e^{- {\lambda'}^2  y }d\lambda' \le  \frac C{\sqrt {|y|}},\\
 \theta(\lambda_R)  {\theta(-y) } \int_{-2|\lambda_R|}^{0}e^{-({\lambda'}^2+2\lambda_R\lambda') y }d\lambda'=   {\theta(-y) } e^{\lambda_R^2y}\int_{- |\lambda_R|}^{|\lambda_R|}e^{-\eta^2 y }d\eta\le  \frac C{\sqrt {|y|}},\\
   \theta(\lambda_R)  {\theta(y) }\int_{-\infty}^{-2|\lambda_R|}
e^{-({\lambda'}^2+2\lambda_R\lambda') y }d\lambda' 
= \theta(\lambda_R)  {\theta(y) } \int_{0}^\infty 
e^{-({\xi}^2+2\lambda_R\xi) y }d\xi  \le  \frac C{\sqrt {|y|}}
\ea
\] 
(via $\eta=\lambda+\lambda_R$, $-\xi=\lambda'+2\lambda_R$), and similar arguments for   terms with $\theta(-\lambda_R)$ factors.

$\underline{\emph{Step 3   (Estimates for $IV$) }}:$  
Let us first claim: for $\lambda=\lambda_R+i\lambda_I=\pm i\corr{\kappa}+se^{i\alpha}$,   
 $0<s<1$, $-\frac \pi 2\le \alpha\le \frac{3\pi}2$, we have
\beq\label{E:residue-1}
\ba{c}\lim_{s\to 0}\int_{-1}^{1}\frac { d\eta}{\eta+i(\lambda_I\mp \corr{\kappa})}=
(\pm i\pi)(2\theta(\corr{\kappa} - |\lambda_{I}|)-1), \\
 \int_{-|\lambda_R|}^{|\lambda_R|}\frac{d\eta}{\eta+i(\lambda_I\mp \corr{\kappa})}
= (\mp i)\ [2\pi(1-\theta(\corr{\kappa} - |\lambda_{I}| )) -2\cot^{-1}\frac {\corr{\kappa}-|\lambda_{I}|}{|\lambda_{R}|}],\ea 
\eeq
where 
\beq\label{E:arccot}
\ba{rl}
\cot^{-1}\frac {\corr{\kappa}-|\lambda_{ I}|}{|\lambda_{ R}|}
=&\left\{
{\begin{array}{lcl}
\frac \pi 2+\alpha,&-\frac \pi 2<\alpha\le\frac \pi 2 ,&0<|\lambda-i\corr{\kappa}|\le 1,\\
\frac {3\pi}2-\alpha,&\frac\pi 2\le\alpha<\frac{3\pi}2,&0<|\lambda-i\corr{\kappa}|\le 1,\\
\frac \pi 2-\alpha,&-\frac \pi 2<\alpha\le\frac \pi 2 ,&0<|\lambda+i\corr{\kappa}|\le 1,\\
-\frac {\pi}2+\alpha,&\frac\pi 2\le\alpha<\frac{3\pi}2,&0<|\lambda+i\corr{\kappa}|\le 1.
\end{array}}
\right.
\ea
\eeq

The  claim is carried out by using the logarithmic functions. Precisely,
\beq\label{E:approach} 
\ba{rl}
&\lim_{s\to 0}\int_{-1}^{1}\frac { d\eta}{\eta+i(\lambda_I\mp \corr{\kappa})}= \log\frac{1+i0^+\sin\alpha}{-1+i0^+\sin\alpha}\\
=&  \left\{
{\begin{array}{l }
+i\pi(2\theta(\corr{\kappa}-|\lambda_I|)-1), \lambda\in\mathbb C^+ ,\\
-i\pi(2\theta(\corr{\kappa}-|\lambda_I|)-1), \lambda\in\mathbb C^-, 
\end{array}}
\right.\\
&{}\\
&\int_{-|\lambda_R|}^{|\lambda_R|}\frac{d\eta}{\eta+i(\lambda_I\mp \corr{\kappa})}=\log\frac{|\lambda_R|+i(\lambda_I\mp \corr{\kappa})}{-|\lambda_R|+i(\lambda_I\mp \corr{\kappa})} 
\\
 =&\left\{
{\begin{array}{l }
\log\frac{|\lambda_R|+i(|\lambda_I|-\corr{\kappa})}{-|\lambda_R|+i(|\lambda_I|-\corr{\kappa})},\ \corr{\kappa}<\lambda_I ,\\
 \log\frac{-|\lambda_R|+i(\corr{\kappa}-|\lambda_I|)}{|\lambda_R|+i(\corr{\kappa}-|\lambda_I|)},\ 0<\lambda_I<\corr{\kappa},\\
 \log\frac{|\lambda_R|+i(\corr{\kappa}-|\lambda_I|)}{-|\lambda_R|+i(\corr{\kappa}-|\lambda_I|)},\ -\corr{\kappa}<\lambda_I<0,\\
 \log\frac{-|\lambda_R|+i(|\lambda_I|-\corr{\kappa})}{|\lambda_R|+i(|\lambda_I|-\corr{\kappa})},\ \lambda_I<-\corr{\kappa}
\end{array}}
\right.
\\
=&\left\{
{\begin{array}{l}
-2i\cot^{-1}\frac{|\lambda_I|-\corr{\kappa}}{|\lambda_R|},\ \corr{\kappa}<\lambda_I ,\\
+2i\cot^{-1}\frac{\corr{\kappa}-|\lambda_I|}{|\lambda_R|},\ 0<\lambda_I<\corr{\kappa},\\
-2i\cot^{-1}\frac{\corr{\kappa}-|\lambda_I|}{|\lambda_R|},\ -\corr{\kappa}<\lambda_I<0,\\
+2i\cot^{-1}\frac{|\lambda_I|-\corr{\kappa}}{|\lambda_R|},\ \lambda_I<-\corr{\kappa} 
\end{array}}
\right.  \\
=&\left\{
{\begin{array}{lc}
-i [2\pi(1-\theta(\corr{\kappa}-|\lambda_I|)) 
-2\cot^{-1}\frac {\corr{\kappa}-|\lambda_{ I}|}{|\lambda_{ R}|}],&\lambda\in\mathbb C^+ ,\\
+i [2\pi(1-\theta(\corr{\kappa}-|\lambda_I|)) 
-2\cot^{-1}\frac {\corr{\kappa}-|\lambda_{ I}|}{|\lambda_{ R}|}],&\lambda\in\mathbb C^-. 
\end{array}}
\right. 
\ea
\eeq
Here we have used $0 < \cot^{-1}x < \pi$, $\cot^{-1}(-x)=\pi -\cot^{-1}x$, and the principal integration.

Plugging \eqref{E:residue-1} into $IV$, one has
\beq\label{E:iv-discrete}
\ba{rl}
&\lim_{s\to 0}   IV\\
=&\pm \frac {i}{\pi}\mathfrak g(x,x',\pm i\corr{\kappa})\ [\theta(y)\lim_{s\to 0}\int_{-1}^{1}\frac { d\eta}{\eta+i(\lambda_I\mp \corr{\kappa})}- 
\lim_{s\to 0}\int_{-|\lambda_R|}^{|\lambda_R|}\frac{d\eta}{\eta+i(\lambda_I\mp \corr{\kappa})}]   \\
=&\mathfrak g(x,x',\pm i\corr{\kappa})\ [\theta(y)-2 +2\theta(-y)\theta(\corr{\kappa}-|\lambda_{I}|) +\frac 2\pi \cot^{-1}\frac {\corr{\kappa}-|\lambda_{I}|}{|\lambda_{R}|}],
\ea
\eeq 
and
\beq\label{E:iv-discrete-new} 
\ba{c}| IV|\le C.\ea
\eeq

$\underline{\emph{Step 4    (Estimates for   $II$)}}:$ By \eqref{E:est-III}, the $L^\infty$-estimate of $II$ reduces to that for 
\[
\ba{rl}
\widetilde {II}=& \widetilde {II}_1+\widetilde {II}_2,\\
\widetilde {II}_1=&\pm \frac {i}{\pi}  \theta(y)(\int_{-1}^{-|\lambda_R|}+\int_{|\lambda_R|}^{1})  \frac { e^{(\lambda_R^2-\eta^2 ) y}-1}{ \eta+i(\lambda_I\mp \corr{\kappa} )} d\eta\\
\widetilde {II}_2=&\mp \frac {i}{\pi}  \theta(-y)\int_{-|\lambda_R|}^ {|\lambda_R|}  \frac {e^{(\lambda_R^2-\eta^2 ) y}- 1  }{ \eta+i(\lambda_I\mp \corr{\kappa} )} d\eta.
\ea
\] 
In case of 
\beq\label{E:ii-l}
\ba{c}
(4a)\ |\lambda_R|\sqrt{|y|}\ge 1,
\ea
\eeq 
by the change of variables $\xi=\frac\eta{|\lambda_R|}$,  \eqref{E:iv-discrete-new}, a Gaussian integration,  and using logarithmic functions,
\beq\label{E:ii-l-1-est}
\ba{rl}
&  | \widetilde {II}_{1} | \\
=  & | \pm \frac {i}{\pi} \theta(y)\{(\int_{-\frac 1{|\lambda_R|}}^{-1}+\int_{1}^{\frac 1{|\lambda_R|}})  \frac { e^{ \lambda_R^2y(1-\xi^2 )}   }{ \xi +i\frac{\lambda_I\mp \corr{\kappa}}{|\lambda_R|}} d\xi -(\int_{-1}^{-|\lambda_R|}+\int_{|\lambda_R|}^{1})  \frac {  1}{ \eta+i(\lambda_I\mp \corr{\kappa} )} d\eta\}|\\
\le &C\theta(y)(\int_{-\frac 1{|\lambda_R|}}^{-1}+\int_{1}^{\frac 1{|\lambda_R|}})\left|\frac { e^{ \lambda_R^2y(1-\xi^2 )}   }{ \xi +i\frac{\lambda_I\mp \corr{\kappa}}{|\lambda_R|}}\right| d\xi+C \\
\le &C(\int_{-\frac 1{|\lambda_R|}}^{-1}+\int_{1}^{\frac 1{|\lambda_R|}})   e^{  1-\xi^2 }     d\xi+C=Ce\sqrt {e^{-1}-e^{-\frac 1{|\lambda_R|^2}}}+C\le C.
\ea
\eeq  Besides, by the change of variables $\omega=\eta\sqrt{|y|}$ and \eqref{E:ii-l},
\beq\label{E:ii-l-2-est}
\ba{rl}
 & | \widetilde {II}_{2} | \\
= &  | \mp \frac {i}{\pi}\theta(-y) \int_{-|\lambda_R|\sqrt{|y|}}^ {|\lambda_R|\sqrt{|y|}}   \frac {  e^{ \lambda_R^2y-\omega^2  }-1}{ \omega+i(\lambda_I\mp  {\kappa} )\sqrt{|y|}} d\omega|\\
= &   | \mp \frac {i}{\pi} \theta(-y) \int_{-1}^ {1}   \frac {  e^{ \lambda_R^2y-\omega^2  }-e^{ \lambda_R^2y-1}}{ \omega+i(\lambda_I\mp  {\kappa} )\sqrt{|y|}} d\omega\mp \frac {i}{\pi}  \theta(-y) \int_{-1}^ {1}   \frac {  e^{ \lambda_R^2y-1}}{ \omega+i(\lambda_I\mp  {\kappa} )\sqrt{|y|}} d\omega \\
&\mp \frac {i}{\pi}\theta(-y) (\int_{-|\lambda_R|\sqrt{|y|}}^{-1}+\int_1^ {|\lambda_R|\sqrt{|y|}}  ) \frac {  e^{ \lambda_R^2y-\omega^2  } }{ \omega+i(\lambda_I\mp  {\kappa} )\sqrt{|y|}} d\omega\\
&\pm \frac {i}{\pi}  \theta(-y) \int_{-|\lambda_R| }^{ |\lambda_R| }    \frac { 1}{ \eta+i(\lambda_I\mp \corr{\kappa} )} d\eta|

\ea
\eeq  
Therefore, applying the mean value theorem, \eqref{E:ii-l}, the logarithmic functions, a Gaussian integration,  and  \eqref{E:iv-discrete-new}, 
\beq\label{E:ii-l-2-est-cont}
\ba{rl}
 & | \widetilde {II}_{2} | \\
 \le& C  \theta(-y) \int_{0}^ {1}   |\frac {  e^{ \lambda_R^2y-\omega^2  }-e^{ \lambda_R^2y-1}}{ \omega-1} |d\omega +C |\int_{-1}^ {1}   \frac { 1}{ \omega+i(\lambda_I\mp  {\kappa} )\sqrt{|y|}} d\omega|\\
& +C\theta(-y) (\int_{-|\lambda_R|\sqrt{|y|}}^{-1}+\int_1^ {|\lambda_R|\sqrt{|y|}}  ) e^{ \lambda_R^2y-\omega^2  } d\omega +C\\
\le &C.
\ea
\eeq

Instead of \eqref{E:ii-l}, we now consider
\beq\label{E:ii-s}
\ba{c}
(4b)\ |\lambda_R|\sqrt{|y|}\le 1.
\ea
\eeq Therefore, by the change of variables $\xi=\frac\eta{|\lambda_R|}$ and the mean value theorem,
\beq\label{E:ii-s-1-est}
\ba{l}
   | \widetilde {II}_{2} |  
=   | \mp \frac {i}{\pi} \theta(-y) \int_{-1}^{-1}  \frac { e^{ \lambda_R^2y(1-\xi^2 )}-1}{ \xi +i\frac{\lambda_I\mp \corr{\kappa}}{|\lambda_R|}} d\xi  |\\
\le  C\theta(-y) \int_{-1}^{-1} \left| \frac { e^{ \lambda_R^2y(1-\xi^2 )}-1}{ \xi -|\lambda_R|\sqrt{|y|}}\right| d\xi  \le C.
\ea
\eeq 
On the other hand, by the change of variables $\omega=\eta\sqrt{|y|}$ and \eqref{E:ii-s},
\beq\label{E:reduced-ii-new}
\ba{rl}
 &| \widetilde {II}_{1} | 
=  |  \pm\frac i\pi\theta(y)(\int_{-1}^{-|\lambda_R|\sqrt{y}}+\int_{|\lambda_R|\sqrt{y}}^{1}) \frac {e^{\lambda_R^2y-\omega^2 }-1}{\omega+i(\lambda_I\mp\corr{\kappa})\sqrt{y}}  d\omega \\
&\pm\frac i\pi \theta(y)(\int_{-\sqrt{y}}^{-1}+\int_{1}^{\sqrt{y}})\frac { e^{\lambda_R^2y -\omega^2 }}{ \omega+i(\lambda_I\mp\corr{\kappa})\sqrt{ y } } d\omega \corr{\mp 
\frac i \pi\theta(y)(\int_{-1}^{-\frac 1{\sqrt{y}}}+\int_{\frac 1{\sqrt{y}}}^{1})\frac {1}{\eta+i(\lambda_I\mp\corr{\kappa})} d\eta} |.
\ea
\eeq
By the mean value theorem, a Gaussian integral, \eqref{E:ii-s}, and the logarithmic functions, 
\beq\label{E:the-mean-value-theorem}
\ba{ l}
   |\widetilde {II}_{ 1}|  
\le  C \theta(y) \int_{|\lambda_R|\sqrt{y}}^{1}  \left|\frac {e^{\lambda_R^2y-\omega^2 }-1}{\omega-|\lambda_R|\sqrt{y}} \right| d\omega +C\theta(y)e^{\lambda_R^2y}\sqrt{e^{-1}-e^{-y}}  +C\le C.
\ea
\eeq
Combining cases (4a) and (4b), we obtain 
\beq\label{E:2-est}
\ba{c}
| \widetilde {II}_{2} |\le  C.
\ea
\eeq

$\underline{\emph{Step 5    (Estimates for $I$)}}:$ By the estimate  of $III$ and $II$, the estimate of $I$ reduces to considering   $\widetilde I=\widetilde I_{ 1}+\widetilde I_{ 2} $, with
\beq\label{E:1-1-bdd}
\ba{rl}
\widetilde {I}_{1} =& (\int_{-1}^{-|\lambda_R|}+\int_{|\lambda_R|}^{1})\frac {  e^{i\eta(x'-x-2\lambda_Iy)} -1     }{ \eta+i(\lambda_I\mp\corr{\kappa}  )} d\eta\\
=&  \int_{-1}^ 1 \frac {  e^{i\eta(x'-x-2\lambda_Iy)} -1     }{ \eta+i(\lambda_I\mp\corr{\kappa}  )} d\eta-   \int_ {-|\lambda_R|}^{|\lambda_R|}\frac {  e^{i\eta(x'-x-2\lambda_Iy)} -1     }{ \eta+i(\lambda_I\mp\corr{\kappa}  )} d\eta ,\\
\widetilde {I}_{2} =&\int_ {-|\lambda_R|}^{|\lambda_R|}\frac {  e^{i\eta(x'-x-2\lambda_Iy)} -1     }{ \eta+i(\lambda_I\mp\corr{\kappa}  )} d\eta. 
\ea
\eeq 
Observe, in either case
\beq\label{E:residue-domain}
\ba{l}
(5a)\,(\lambda_I\mp \kappa)(x'-x-2\lambda_Iy)> 0,\\
(5b)\,(\lambda_I\mp \kappa)(x'-x-2\lambda_Iy)\le 0,\textit{ and }\left|\, |\lambda_R|-|\lambda_I\mp\corr{\kappa} |\,\right|\ge \frac 12|\lambda_R|,
\ea
\eeq   by the residue theorem \cite[\S 6, Chapter III]{C63},
\beq\label{E:deform-residue}
\ba{rl}
&|\int_ {-|\lambda_R|}^{|\lambda_R|}\frac {  e^{i\eta(x'-x-2\lambda_Iy)} -1     }{ \eta+i(\lambda_I\mp\corr{\kappa}  )} d\eta|\\
\le &  |\int_{-\frac \pi 2}^{\frac{3\pi}2}\theta([x'-x-2\lambda_Iy]\sin\beta)
\frac {  e^{i|\lambda_R|(\cos\beta+i\sin\beta)(x'-x-2\lambda_Iy)} -1     }{ |\lambda_R|e^{i\beta}+i (\lambda_I\mp\corr{\kappa}  )}|\lambda_R|ie^{i\beta} d\beta|\\
&+2\pi \theta(-(\lambda_I\mp \kappa)(x'-x-2\lambda_Iy))|e^{(\lambda_I\mp \kappa)(x'-x-2\lambda_Iy)}-1|\\
\le & C\int_{-\frac \pi 2}^{\frac{3\pi}2}\theta([x'-x-2\lambda_Iy]\sin\beta)\\
&\times |  e^{i |\lambda_R|(x'-x-2\lambda_Iy)\cos\beta -|\lambda_R|(x'-x-2\lambda_Iy) \sin\beta} -1      | d\beta  +C\\
\le & C.
\ea
\eeq Instead of \eqref{E:residue-domain}, if
\beq\label{E:right}
\ba{c}
(5c)\ (\lambda_I\mp \kappa)(x'-x-2\lambda_Iy)\le 0, \textit{ and }\left|\, |\lambda_R|-|\lambda_I\mp\corr{\kappa} |\,\right|\le \frac 12|\lambda_R| 
\ea
\eeq holds, we will show the integral $\widetilde I_2$ is basically a Hilbert transform. Precisely, by deforming the contours \cite[\S 6, Chapter III]{C63},
\beq\label{E:deform-contour}
\ba{rl}
&|\int_ {-|\lambda_R|}^{|\lambda_R|}\frac {  e^{i\eta(x'-x-2\lambda_Iy)} -1     }{ \eta+i(\lambda_I\mp\corr{\kappa}  )} d\eta|\\
\le &|\int_{\Gamma_-}\frac {  e^{i\eta(x'-x-2\lambda_Iy)} -1     }{ \eta+i(\lambda_I\mp\corr{\kappa}  )} d\eta|+|\int_{\Gamma_+}\frac {  e^{i\eta(x'-x-2\lambda_Iy)} -1     }{ \eta+i(\lambda_I\mp\corr{\kappa}  )} d\eta|\\
&+|\int_{\Gamma}\frac {  e^{i\eta(x'-x-2\lambda_Iy)} -1     }{ \eta+i(\lambda_I\mp\corr{\kappa}  )} d\eta|+2\pi |e^{(\lambda_I\mp \kappa)(x'-x-2\lambda_Iy)}-1|
\ea
\eeq where
\beq\label{E:gamma-pm-m}
\ba{l}
\Gamma_-=\{\eta=-|\lambda_R|-i(\lambda_I\mp\corr{\kappa}  )t: \ \ 0\le t\le 1\},\\
\Gamma_+=\{\eta=+|\lambda_R|-i(\lambda_I\mp\corr{\kappa}  )t: \ \ 0\le t\le 1\},\\
\Gamma=\{\eta=t- i(\lambda_I\mp\corr{\kappa}  ) : \ \ -|\lambda_R|\le t\le |\lambda_R|\}.
\ea
\eeq
Due to \eqref{E:right} and \eqref{E:gamma-pm-m}, one has
\beq\label{E:deform-contour-pm}
\ba{rl}
 &|\int_{\Gamma_\pm}\frac {  e^{i\eta(x'-x-2\lambda_Iy)} -1     }{ \eta+i(\lambda_I\mp\corr{\kappa}  )} d\eta|\\
 =&|\int_0^1\frac { ^{\pm i|\lambda_R|(x'-x-2\lambda_Iy)+(\lambda_I\mp\corr{\kappa}  )(x'-x-2\lambda_Iy)t} -1    }{ \pm|\lambda_R|+i(\lambda_I\mp\corr{\kappa})(1-t  )} (\lambda_I\mp\corr{\kappa})dt|\\
 \le & \int_0^1|e^{\pm i|\lambda_R|(x'-x-2\lambda_Iy)+(\lambda_I\mp\corr{\kappa}  )(x'-x-2\lambda_Iy)t} -1|dt\le C.
\ea
\eeq On the other hand,
\beq\label{E:gamma-kappa}
\ba{rl}
&|\int_{\Gamma}\frac {  e^{i\eta(x'-x-2\lambda_Iy)} -1     }{ \eta+i(\lambda_I\mp\corr{\kappa}  )} d\eta|\\
=&|\int_ {-|\lambda_R|}^{|\lambda_R|}\frac {  e^{ (\lambda_I\mp\corr{\kappa}  )(x'-x-2\lambda_Iy)+it(x'-x-2\lambda_Iy)} -1     }{ t  } dt|\\
\le& |\int_ {-|\lambda_R|}^{|\lambda_R|}\frac {  e^{ it(x'-x-2\lambda_Iy)}       }{ t  } dt|\\
\le &C
\ea
\eeq
by using symmetries and the residue theorem \cite[\S 6, Chapter III]{C63}. 
Therefore, combining cases (5a), (5b), and (5c), $|\widetilde {I}_{2} | \le C$. The same method can be adapted to proving $| \int_{-1}^ 1 \frac {  e^{i\eta(x'-x-2\lambda_Iy)} -1     }{ \eta+i(\lambda_I\mp\corr{\kappa}  )} d\eta|<C$ which yields 
\beq\label{E:hilbert-bdd}
\ba{c}
 |  \widetilde {I}(x,x',y,\lambda)|\le C 
\ea
\eeq is justified. Combining with results from previous steps, we obtain 
estimate \eqref{E:eigen-green} 
 in case of $|\lambda_R|\le 1$.

$\underline{\emph{Step 6    (Estimates for $I^\dagger$, $II^\dagger$, $III^\dagger$)}}:$ 
For $II^\dagger$, we adapt the argument for (4a) in $\widetilde{II}$ in $\underline{\emph{Step 4}}$ and for $\widetilde{I}$ in $\underline{\emph{Step 5}}$; for $I^\dagger$, $III^\dagger$, we apply Gaussian type estimates \eqref{E:gaussian}.
\end{proof}

Green functions $G$ and $\tilde G$, defined by \eqref{E:sym-0} - \eqref{E:chi-a},  can be extended to $y\ne 0$ and $\lambda\ne \pm i$ by the following lemma.
\begin{lemma}\label{L:green-local-cont} 
For $\forall y\ne 0$, the Green function $  G$, defined by \eqref{E:sym-0} - \eqref{E:chi-a},
\[\ba{c}
G(x,x',y,\lambda_R+i0^+) =  G(x,x',y,\lambda_R+i0^-) , \ \forall\lambda_R ;\\
G(x,x',y,0^++i\lambda_I) = G(x,x',y,0^-+i\lambda_I),\ \forall\lambda_I\ne \pm \corr{\kappa};\\
G(x,x',y,\lambda_R+i(\pm \corr{\kappa}+0^+)) =  G(x,x',y,\lambda_R+i(\pm\corr{\kappa}+0^-)) , \ \forall\lambda_R\ne 0.
\ea\]
\end{lemma}
\begin{proof}  Follow  from the proof of previous lemma.

\end{proof}

\begin{lemma}\label{L:Riemann-Lebesque}
Suppose $f\in L^1\cap L^\infty$. For $\forall \lambda\ne\pm i\corr{\kappa}$, the Green function $\widetilde G$, defined by \eqref{E:sym-0} - \eqref{E:chi-a}, satisfies  
\[
\ba{c}\mbox{$\widetilde G\ast f\to 0$ uniformly as $|x|,\,|y|\to\infty$.}
\ea\]
Here the $\ast$ operator is defined by 
\beq\label{E:ast}
\ba{c}\widetilde G\ast f(x,y,\lambda)=\iint\widetilde G(x,x',y-y',\lambda)f(x',y')dx'dy'.
\ea\eeq
\end{lemma}

\begin{proof} Since $\lambda\ne\pm i\corr{\kappa}$ fixed, both  $x$-asymptotics   and $y$-asymptotics can be obtained by Lemma \ref{L:eigen-green},   $f\in {L^1\cap L^\infty}$, and  the dominated convergence theorem.
\end{proof}

\begin{proof}\textbf{\textit{of Theorem \ref{T:KP-eigen-existence}. }} From   Lemma \ref{L:eigen-green} and $v_0\in L^1\cap L^\infty$, for $\lambda\ne \pm i\corr{\kappa}$,   if $f(x,y)\in L^\infty(\mathbb R^2)$, then the map $f\mapsto \widetilde G\ast v_0f$ is bounded from $L^\infty(\mathbb R^2)$ to $L^\infty(\mathbb R^2)$ and has a norm bounded by $ C|v_0|_{L^1\cap L^\infty}$. Consequently, if $ {|v_0|_{L^1\cap L^\infty}}\ll 1$, then for $\lambda\ne \pm i\corr{\kappa}$ the integral equation
\begin{equation}\label{E:1-ode}
 m=\vartheta_--\widetilde G\ast v_0m
\end{equation} is uniquely solved for $m\in L^\infty(\mathbb R^2)$.   Moreover,  $\partial_y^j\partial_x^k m(x,y,\lambda)\in  L^\infty(\mathbb R^2)$,   $0\le j,\,k\le 2$, can be proved via the formula
\[
\ba{rl}
& \partial_x  m(x,y,\lambda)\\
=&(1-\widetilde G\ast v_0)^{-1}\left[ \partial_x  \widetilde G\right]\ast v_0 (1-\widetilde G\ast v_0)^{-1}\vartheta_-\\
=&(1-\widetilde G\ast v_0)^{-1}\left[ \partial_x  \widetilde G\right]\ast v_0 m,
\ea
\] integration by parts, $\partial_y^j\partial_x^kv_0\in L^1\cap L^\infty$, $0\le j,\,k\le 2$, and an induction argument (Here we make remarks on $\partial_x \widetilde G$. To take $x$-derivatives of exponential terms, we need to use the antisymmetries in $x$, $x'$  and apply integration by parts). 
From Lemma \ref{L:Riemann-Lebesque},  $\Psi(x,y,\lambda)=m(x,y,\lambda)e^{(-i\lambda) x+(-i\lambda)^2y}$ is a solution  to the problem \eqref{E:kp-line-normal}, (\ref{E:kp-line-normal-bdry-1}).

\end{proof}

\section{The forward problem II: the forward scattering transform}\label{S:IST-kp-sd}

The scattering data can be defined by the non-holomorphic part of the eigenfunctions, i.e., $\overline\partial m=\partial_{\bar\lambda}m$ \cite{BC81}, \cite{FA83}, \cite{ABF83}.   Classically, $\overline\partial m$ can often be computed in terms of $m$ by noting that both $m$ and $\overline\partial m$ satisfy the same equation. Thus $m$ can be reconstructed from a knowledge of this relationship by solving a $\overline\partial$-problem, namely, a Cauchy integral equation \cite{BC81}, \cite{W85}, \cite{W87}. The main goals of this section are to compute $\overline\partial m$, to define the scattering transform and to characterize its algebraic and analytical constraints, and to formulate a Cauchy integral equation.

\subsection{Discrete scattering data of $m(x,y,\lambda)$}  

Major distinction of non-localized KPII equation from other integrable systems is the occurrence of non-meromorphic  discrete scattering data. In the following lemma, we will prove the discontinuities  at $\lambda=\pm i\corr{\kappa}$ of the Green function $\widetilde G$, defined by \eqref{E:sym-0} - \eqref{E:chi-a}, which yield the non-meromorphic properties of the discrete scattering data.

\begin{definition}\label{D:terminology} 
For $z\in \mathfrak Z=\{\pm i\corr{\kappa},\, \iota\}$, $\iota =  2i\corr{\kappa}$, define 
\beq\label{E:dz}
\ba{c}
D_{z}=\{\lambda\in\CC\, :\, |\lambda-z|< 1\},\ 
   D_z^\times=\{\lambda\in\CC\, :\, 0<|\lambda-z|<1\} ;\\
D_{z,a}=\{\lambda\in\CC\, :\, |\lambda-z|< a\},\ 
   D_{z,a}^\times=\{\lambda\in\CC\, :\, 0<|\lambda-z|< a\} ,   
 \ea
 \eeq  and characteristic functions  
\beq\label{E:chi}
\begin{array}{c}
\textit{ $E_z(\lambda)\equiv 1$ on $D_z$, $E_z (\lambda)\equiv 0$ elsewhere};\\
\textit{ $E_{z,a}(\lambda)\equiv 1$ on $D_{z,a}$, $E_{z,a} (\lambda)\equiv 0$ elsewhere}.  
\end{array}
\eeq Moreover, define the polar coordinate for   $D_{z,a}^\times$ to be $\{(s,\alpha)|\lambda=z+se^{i\alpha},\ 0<s<a,\,-\frac \pi 2\le\alpha \le\frac{3\pi}2\}$.
     
\end{definition}

\begin{lemma}\label{L:green-local} 
For $y\ne 0$, $\lambda=\lambda_R+i\lambda_I\in D_{\pm i\corr{\kappa}}^\times$, 
\beq\label{E:g-asym-pm-i-prep-limit}
\begin{array}{c}
 \widetilde G(x,x',y,\lambda)
= {\mathfrak G}_\pm(x,x',y)+\frac 2\pi
{\mathfrak g(x,x',\pm i\corr{\kappa})\cot^{-1}\frac{\corr{\kappa}-|\lambda_I|}{|\lambda_R|}+\omega_\pm(x,x',y,\lambda)},  
\end{array}
\eeq with
\beq\label{E:g-asym-pm-i-prep-limit-1}
\ba{rl}
 {\mathfrak G}_\pm(x,x',y)
 =  \int_{|\eta|\ge 1}   \theta(y )\frac{ e^{-y{\eta}^2+i\eta(x'-x\mp 2\corr{\kappa}y)}\mathfrak g(x,x',\eta\pm i\corr{\kappa}) }{2\pi a(\eta\pm i\corr{\kappa})}d\eta & 
 \\
 +  \int_{|\eta|\le 1}  \theta(y )\frac{ e^{-y{\eta}^2+i\eta(x'-x\mp 2\corr{\kappa}y)}\mathfrak g(x,x',\eta\pm i\corr{\kappa})- \mathfrak g(x,x', \pm i\corr{\kappa})}{2\pi a(\eta\pm i\corr{\kappa})}d\eta &\\
 +\mathfrak g(x,x',\pm i\corr{\kappa})\ [\theta(y)(1+\frac 1\pi)-2]&\in\RR,
 \ea\eeq 
and  
\beq\label{E:omega-definition-new}
\ba{rl}
&\omega_\pm(x,x',y,\lambda)\\
&=\theta(y) \int_{|\eta|\ge 1} [\frac{e^{-i\lambda_R(x'-x-2\lambda_Iy)}\tilde{\mathfrak F}(x,x',y,\lambda,\eta) }{2\pi a(\eta+i\lambda_I)}-\frac{\tilde{\mathfrak F}(x,x',y,\pm i\corr{\kappa},\eta) }{2\pi a(\eta\pm i \corr{\kappa})}  ]d\eta 
\\
&+\int_{|\eta|\le 1} \Xi(y,\lambda,\eta) e^{-i\lambda_R(x'-x-2\lambda_Iy)}  \frac{ \tilde{\mathfrak F}(x,x',y,\lambda,\eta )-  \mathfrak g(x,x', \pm i\corr{\kappa})}{{2\pi}a(\eta+i\lambda_I)} d\eta\\
&-\int_{|\eta|\le 1}  d\eta\ \theta(y)   \frac{ e^{ -\eta^2 y+i\eta(x'-x\mp 2\corr{\kappa} y)}\mathfrak g(x,x',\eta\pm i\corr{\kappa})- \mathfrak g(x,x', \pm i\corr{\kappa})}{{2\pi}a(\eta\pm i\corr{\kappa})}\\
&+\int_{|\eta|\le 1} \Xi(y,\lambda,\eta)\frac {[e^{-i\lambda_R(x'-x-2\lambda_Iy)}-1] \mathfrak g(x,x',\pm i\corr{\kappa}) }{2\pi a(\eta+i\lambda_I)}d\eta -\frac {\mbox{sgn}(\lambda_R)\lambda_R}{ \pi}\mathfrak g(x,x',\pm i\corr{\kappa}) \\
&-2\theta(-y)\theta(\corr{\kappa}^2- \lambda_I^2) [e^{(\lambda^2+\corr{\kappa}^2)y {+i\lambda(x-x')}\pm\corr{\kappa}(x-x')}-1]\mathfrak g(x,x'\pm i\corr{\kappa})\\
&\pm\frac i\pi\theta(y)\mathfrak g(x,x',\pm i\corr{\kappa}) [\log\frac{1+is\sin\alpha}{-1+is\sin\alpha}-\log\frac{1+i0^+\sin\alpha}{-1+i0^+\sin\alpha} ],
\ea
\eeq
where $a(\lambda)$, $\mathfrak g(x,x',\lambda)$, $\cot^{-1}\frac{\corr{\kappa}-|\lambda_I|}{|\lambda_R|}$, $\Xi(y,\lambda,\eta)$, and $\tilde{\mathfrak F}(x,x',y,\lambda,\eta)$ are defined by \eqref{E:chi-a}, \eqref{E:g-def}, \eqref{E:arccot},   \eqref{E:f}, and $(s,\alpha)$ denotes the polar coordinates for $D_{\pm i\corr{\kappa}}^\times$. Moreover, 
\begin{gather}
\ba{c}
 |\omega_\pm  |_{L^\infty(D_{\pm i\corr{\kappa}})}\le C ,\ |\frac \partial{\partial s}\omega_\pm  |_{L^\infty(D_{\pm i\corr{\kappa}})}\le C(1+|x-x'\pm 2\corr{\kappa}y|).
\ea \label{E:green-variation-derivative}
\end{gather} 
\end{lemma}
\begin{proof} $\underline{\emph{Step 1 } \textit{(Proof of \eqref{E:g-asym-pm-i-prep-limit} - \eqref{E:omega-definition-new})}} :$ Let $\lambda_0$ be a radial limit $\lambda$ at $\pm i\corr{\kappa}$. Write 
\[
\ba{c}
\widetilde G=e^{ -i\lambda_R(x'-x-2\lambda_Iy)}(I+II+III+IV+V+ \widetilde G_d^\flat),\\
\widetilde G_d^\flat=-2\theta(-y)\theta(\kappa-|\lambda_I|)\mathfrak g(x,x',\pm i\kappa)e^{(\lambda_I\mp \kappa)(x'-x)}e^{[\lambda_R^2+(\kappa^2-\lambda_I^2)]y }
\ea
\]where $I$-$V$ are defined as in the proof of Lemma \ref{L:eigen-green}. By the 
dominated convergence theorem \cite[\S 6, Chapter III]{C63}, 
\[
\ba{c}
\textit{$ {III}$, $ {V}$, $\widetilde {II}$ (cf. \eqref{E:ii-s-1-est}, \eqref{E:reduced-ii-new} in case (4b)), $ \widetilde {I}$ }\\
\textit{are continuous at $\lambda=\pm i\corr{\kappa}$, for  fixed $y\ne 0$, $x\in\RR$.} 
\ea
\]  Together with \eqref{E:iv-discrete} and 
\[
\ba{c}
\widetilde G_d(x,x',y,\lambda_0)=\widetilde G_d^\flat(x,x',y,\lambda_0)=-2\theta(-y)\theta(\corr{\kappa}-|\lambda_{0,I}|)\mathfrak g(x,x',\pm i\corr{\kappa}),
\ea\]  we then prove $\widetilde G$ can be written as \eqref{E:g-asym-pm-i-prep-limit} where $\mathfrak G_\pm$, $\omega_\pm$ are defined by \eqref{E:g-asym-pm-i-prep-limit-1} and \eqref{E:omega-definition-new}.

 $\underline{\emph{Step 2 } \textit{(Proof of \eqref{E:green-variation-derivative})}} :$ 
The first inequality in \eqref{E:green-variation-derivative} follows from Lemma \ref{L:eigen-green}. For simplicity, we only give the proof of the second inequality at $s=0$ (i.e. $\lambda=\lambda_0$) because the computation is similar for $s\ne 0$. From $\underline{\emph{Step 4}}$  of the proof for Lemma \ref{L:eigen-green},
\beq\label{E:IV-derivative}
\ba{rl}
&\lim_{s\to  0}\frac {(IV+\widetilde G_d^\flat)(x,x',y,\lambda)-(IV+\widetilde G_d^\flat)(x,x',y,\lambda_0) }{s}\\
=&\frac{\partial}{\partial s}|_{s=0}\{-2\theta(-y)\theta(\corr{\kappa}-|\lambda_{0,I}|) [ e^{(\lambda_I\mp \kappa)(x'-x)}e^{[\lambda_R^2+(\kappa^2-\lambda_I^2)]y }-1]\\
&\times\mathfrak g(x,x'\pm i \corr{\kappa}) 
 \pm\frac i \pi \theta(y)[\log\frac{1+is\sin\alpha}{-1+is\sin\alpha}-\log\frac{1+i0^+\sin\alpha}{-1+i0^+\sin\alpha}]\mathfrak g(x,x',\pm i \corr{\kappa})\}\\
=&\mathfrak g(x,x',\pm i\corr{\kappa})\{(x'-x\mp 2\corr{\kappa}y)  \sin\alpha\ \widetilde G_d^\flat(x,x',y,\lambda_0)\\
&\mp\frac 2\pi\theta(y) \sin\alpha \mathfrak g(x,x',\pm i\corr{\kappa})\}.
\ea
\eeq
Furthermore, by 
\[
\begin{array}{c}
 \frac {\partial f}{\partial s}= \cos\alpha\frac {\partial f}{\partial \lambda_R}+\sin\alpha\frac {\partial f}{\partial \lambda_I} ,\ 
\frac{\partial}{\partial s}\frac{1}{ \eta+i(\lambda_I\mp\kappa)}=i\sin\alpha \frac \partial{\partial{\eta}}\frac{1}{\eta+i(\lambda_I\mp\kappa)}, \\
\frac{\partial}{\partial \eta}\chi_{[-1,1]}(\eta)\Xi(y,\lambda,\eta)=\theta(y)(\delta_{\eta= 1 }-\delta_{\eta=- 1})-\mbox{sgn}(\lambda_R)(\delta_{\eta= \lambda_R }-\delta_{\eta=- \lambda_R }), \\
\frac{\partial}{\partial s}\Xi(y,\lambda,\eta)=-\mbox{sgn}(\lambda_R)\cos\alpha(\delta_{\eta= \lambda_R }+\delta_{\eta=- \lambda_R }),\\
\textit{$\chi_A$,  the characteristic function of the set $A$},
\end{array}
\]
and integration by parts, 
\beq\label{E:decomp-omega-ds}
\ba{rl}
&\lim_{s\to 0} \frac{(I+II+III)(x,x',y,\lambda)-(I+II+III)(x,x',y,\lambda_0)}s\\
=&\pm\frac i\pi \frac {\partial}{\partial s} |_{s=0}\int_{-1}^1 \frac {\Xi(y,\lambda,\eta)[ \tilde{\mathfrak F}(x,x',y,\lambda,\eta)- \mathfrak g(x,x',\pm i\corr{\kappa}) ]}{ \eta +i(\lambda_I\mp\kappa)}d\eta \\
=& \pm\frac i\pi\lim_{s\to 0}\{\int_{-1}^1\frac {\Xi(y,\lambda,\eta)}{\eta+i(\lambda_I\mp\kappa)}(\frac {\partial}{\partial s}-i\sin\alpha \frac{\partial}{\partial \eta})\tilde{\mathfrak F}(x,x',y,\lambda,\eta)d\eta\\
&\pm\frac 1\pi\theta(y)\sin\alpha\int_{-1}^1\frac{\tilde{\mathfrak F}(x,x',y,\lambda,\eta)- \mathfrak g(x,x',\pm i\corr{\kappa})}{\eta +i(\lambda_I\mp\kappa)}[\delta_{\eta=1}-\delta_{\eta=-1}]d\eta   
\\
&\mp\frac 1\pi\mbox{sgn}(\lambda_R)\sin\alpha\int_{-1}^1\frac{\tilde{\mathfrak F}(x,x',y,\lambda,\eta)- \mathfrak g(x,x',\pm i\corr{\kappa})}{\eta +i(\lambda_I\mp\kappa)} [\delta_{\eta=\lambda_R}-\delta_{\eta=-\lambda_R}]  d\eta \}.
\ea
\eeq
Since
\[
\ba{rl}
\blacktriangleright\quad  &\lim_{s\to 0}(\frac {\partial}{\partial s}
-i\sin\alpha \frac{\partial}{\partial \eta})\tilde{\mathfrak F}(x,x',y,\lambda,\eta) \\
=&\lim_{s\to 0}\{\cos\alpha\frac{\partial}{\partial{\lambda_R}} e^{\lambda_R^2y+(-\eta^2-2i\lambda_I\eta) y+i\eta(x'-x)} \mathfrak g(x,x',\eta+i\lambda_I) \\
&+  \sin\alpha\frac \partial{\partial{\lambda_I}} e^{ \lambda_R^2 y +(-\eta^2-2i\lambda_I\eta) y+i\eta(x'-x)} \mathfrak g(x,x',\eta+i\lambda_I)
\ea
\]
\beq
\ba{rl}\label{f-ds}
&-i\sin\alpha    \frac{\partial}{\partial \eta} e^{(-\eta^2\mp 2i\kappa\eta) y+i\eta(x'-x)}\mathfrak g(x,x',\eta\pm i\corr{\kappa})\} \\
=& (x'-x\mp 2\corr{\kappa}y)\sin\alpha\ \tilde{\mathfrak F}(x,x',y,\pm i\corr{\kappa},\eta), \\
\blacktriangleright\quad
&\lim_{s\to 0}
\int_{-1}^1\frac{\tilde{\mathfrak F}(x,x',y,\lambda,\eta)- \mathfrak g(x,x',\pm i\corr{\kappa})}{\eta +i(\lambda_I\mp\kappa)} \delta_{\eta=  \lambda_R }  d\eta   \\
=&\lim_{s\to 0} \frac{e^{i\lambda_R(x'-x-2\lambda_Iy)}\mathfrak g(x,x',\lambda)- \mathfrak g(x,x',\pm i\corr{\kappa})}{\lambda_R +i(\lambda_I\mp\kappa)}\\
 =&   \partial_\lambda[\mathfrak g(x,x',\lambda)]_{s=0}  + \mathfrak g(x,x',\pm i\corr{\kappa})\partial_\lambda[e^{i\lambda_R(x'-x-2\lambda_Iy)}]_{s=0} , \\
\blacktriangleright\quad
&\lim_{s\to 0}\int_{-1}^1\frac{\tilde{\mathfrak F}(x,x',y,\lambda,\eta)-\mathfrak g(x,x',\pm i\corr{\kappa})}{ \eta +i(\lambda_I\mp\kappa)} \delta_{\eta= - \lambda_R }  d\eta 
\\
=&\lim_{s\to 0} \frac{e^{-i\lambda_R(x'-x-2\lambda_Iy)}\mathfrak g(x,x', -\overline\lambda)- \mathfrak g(x,x',\pm i\corr{\kappa})}{-\lambda_R +i(\lambda_I\mp\kappa)}
\\
=& - \partial _{\overline\lambda}[ \mathfrak g(x,x',-\overline\lambda)]_{s=0}-\mathfrak g(x,x',\pm i\corr{\kappa}) \partial _{\overline\lambda}[ e^{-i\lambda_R(x'-x-2\lambda_Iy)}]_{s=0},
\ea
\eeq
we obtain
\beq\label{E:123-ds}
\ba{rl}
&\lim_{s\to 0} \frac{(I+II+III)(x,x',y,\lambda)-(I+II+III)(x,x',y,\lambda_0)}s \\
=&  (x'-x\mp 2\corr{\kappa}y)\sin\alpha(I+II+III+IV)(x,x',y,\lambda_0)\\
&\pm\frac 1\pi\theta(y)\sin\alpha\int_{-1}^1\frac{\tilde{\mathfrak F}(x,x',y,\pm i\corr{\kappa},\eta)-\mathfrak g(x,x',\pm i\corr{\kappa})}{ \eta  }[\delta_{\eta=1}-\delta_{\eta=-1}]d\eta.
\ea
\eeq 
Similarly,
\beq\label{E:5-ds}
\ba{rl}
  &\lim_{s\to 0}\frac{ V(x,x',y,\lambda)-V(x,x',y,\lambda_0)}s\\
   = &(x'-x\mp 2 {\kappa}y)\sin\alpha V(x,x',y,\lambda_0)\\
&\pm\frac 1\pi\theta(y)\sin\alpha  (\int_{-\infty}^{-1}+\int^\infty_1) \frac{\tilde{\mathfrak F}(x,x',y,\pm i\corr{\kappa},\eta)}{\eta }[-\delta_{\eta=1}+\delta_{\eta=-1}]d\eta\\
&-\frac 1{2\pi}\lim_{s\to 0}\int_{-\infty}^\infty\tilde{\mathfrak F}(x,x',y,\pm i\corr{\kappa},\eta) \mbox{sgn}(\lambda_R)\cos\alpha (\delta_{\eta=\lambda_R}+\delta_{\eta=-\lambda_R})d\eta\\
&+\frac i{2\pi}\sin\alpha\theta(y)\int_{-\infty}^\infty\{(-2\eta y)e^{-\eta^2y+i\eta(x'-x\mp 2\kappa y)}\mathfrak g(x,x',\eta\pm i\kappa)\\
& +e^{-\eta^2y+i\eta(x'-x\mp2\kappa y)} \partial_\eta\mathfrak g (x,x',\eta\pm i\kappa)\} d\eta. 
\ea
\eeq Assembling \eqref{E:IV-derivative}, \eqref{E:123-ds}, \eqref{E:5-ds},   we obtain
\beq\label{E:green-variation-derivative-ans}
 \ba{rl}
 &\lim_{s\to  0}\frac { \omega_\pm}{s} 
 =  (x'-x\mp 2\corr{\kappa}y)\sin\alpha\widetilde G (x,x',y,\lambda_0)\\
 &-\frac 1{\pi}\mbox{sgn}(\lambda_{0,R})\cos\alpha\mathfrak g(x,x',\pm i\kappa)\\
&+\frac i{2\pi}\sin\alpha\theta(y)\int_{-\infty}^\infty\{(-2\eta y)e^{-\eta^2y+i\eta(x'-x\mp 2\kappa y)}\mathfrak g(x,x',\eta\pm i\kappa)\\
& +e^{-\eta^2y+i\eta(x'-x\mp2\kappa y)} \partial_\eta\mathfrak g (x,x',\eta\pm i\kappa)\} d\eta. 
\ea 
\eeq Therefore \eqref{E:green-variation-derivative} is justified for $s=0$. The case for $s\ne 0$ can be proved by the same method.

\end{proof}

\begin{theorem}\label{T:m-lambda}  Let $\partial_y^j\partial_x^kv_0\in L^1\cap L^\infty$, $0\le j,\,k\le 2$, $|v_0|_{L^1\cap L^\infty}\ll 1$, $v_0(x,y)\in\RR$, and    $\lambda=\lambda_R+i\lambda_I\in D_{\pm i\corr{\kappa}}^\times$. Then  
\beq\label{E:define-remainder-m}
\ba{c}
  m(x,y,\lambda)=
\left\{ 
{\begin{array}{ll}
 \frac{m_{+,-1}(x,y,\lambda)}{\lambda-i\corr{\kappa}} +m_{+,r}(x,y,\lambda),
&\lambda\in D_{+i\corr{\kappa}}^\times,\\
m_{-,0} (x,y,\lambda) +m_{-,r}(x,y,\lambda), &
\lambda\in D_{-i\corr{\kappa}}^\times,
\end{array}}
\right.
\ea
\eeq 
(cf. \cite[Eq.(5.7), (5.8)]{BP302}) with 
\begin{gather}
\begin{array}{l}
m_{+, -1}(x,y,\lambda)= {\frac {2i\corr{\kappa}\Theta_{+}(x,y)}{1+\gamma_{+} \cot^{-1}\frac{\corr{\kappa}-|\lambda_I|}{|\lambda_R|}}},\\
\Theta_{+}(x,y)=(1+{\mathfrak G}_+\ast v_0 )^{-1}\varphi_+(x,+i\corr{\kappa}),
\\
\gamma_{+} =\frac 2\pi\iint v_0(x,y)\vartheta_+(x,+i\corr{\kappa})\Theta_{+}(x,y)dxdy;
\end{array}
\label{E:gamma}\\
\begin{array}{l}
m_{-,0} (x,y,\lambda)= {\frac {\Theta_{-}(x,y)}{1+\gamma_{-} \cot^{-1}\frac{\corr{\kappa}-|\lambda_I|}{|\lambda_R|}},}
\\
\Theta_{-}(x,y)=(1+{\mathfrak G}_-\ast v_0 )^{-1}\vartheta_-(x,-i\corr{\kappa}),
\\
\gamma_{-} =\frac 2\pi\iint v_0(x,y)\varphi_-(x,-i\corr{\kappa})\Theta_{-}(x,y)dxdy,
\end{array}\label{E:gamma-1}
\end{gather}
where $\varphi_\pm$, $\vartheta_\pm$, ${\mathfrak G}_\pm$, $\cot^{-1}\frac{\corr{\kappa}-|\lambda_I|}{|\lambda_R|}$ are defined by \eqref{E:eigen-x},   \eqref{E:g-asym-pm-i-prep-limit}, \eqref{E:arccot}, and
\beq\label{E:m-pm-i-0-+} 
 \ba{c}  {\left|  m_{+,r}(x,y,\lambda)  \right|_{L^\infty(D_{+i\corr{\kappa}})}
  \le  C (1+|x|+|y| )} , \\ \left|(\lambda-i\corr{\kappa})  m_{+,r}(x,y,\lambda)  \right|_{L^\infty(D_{\pm i\corr{\kappa}})}
  \le  C   ,  \\
  \left|  m_{-,r}(x,y,\lambda)  \right|_{L^\infty(D_{- i\corr{\kappa}})}
  \le  C   ,  
    \   m_{-,r}(x,y,-i\corr{\kappa}+0^+e^{i\alpha})  =0   , \\
    \left|\frac\partial{\partial s}  m_{-,r}(x,y,\lambda)  \right|_{L^\infty(D_{-i\corr{\kappa}})}
  \le  C (1+|x|+|y| ),
\ea 
\eeq
\beq
\ba{c} \gamma_+=\gamma_-\ \in\mathbb R,\quad\quad
|\gamma_\pm|\le C|v_0|_{{L^1}}.\ea\label{E:gamma-pm-est}
\eeq
 Moreover,    
\begin{gather}
\ba{c}
m(x,y,\lambda)=\overline{m(x,y,-\overline\lambda)},  \quad \textit{ for }\lambda\ne\pm i\corr{\kappa},
\ea\label{E:m-reality-1}\\ 
\ba{c}
m_{+,-1}(x,y,i\corr{\kappa}+0^+e^{i\alpha}) 
=   s_d e^{-2\corr{\kappa}x} m _{-,0}(x,y,-i\corr{\kappa}+0^+e^{i(\pi+\alpha)}) \  \in i\mathbb R ,
\ea\label{E:residue} 
\end{gather} 
 with the {\bf normalization constant} $  s_d=2i\corr{\kappa}$.
 
\end{theorem}

\begin{proof} $\underline{\emph{Step 1 (Proof for \eqref{E:define-remainder-m} - \eqref{E:m-pm-i-0-+})}}:$ To prove the lemma, we will compute the leading terms of $m(x,y,\lambda)$ at $\pm i \corr{\kappa}$. Denote 
\beq\label{E:p-definition}
\ba{l} \wp_\pm(x,x',y,\lambda)  
=   1+[{\mathfrak G}_\pm +\frac {2}\pi\mathfrak g(x,x',\pm i\corr{\kappa})\cot^{-1}\frac{\corr{\kappa}-|\lambda_I|}{|\lambda_R|}
]\ast v_0
\ea
\eeq 
as the leading terms in $(1+ \widetilde G\ast v_0)$ at $\lambda=\pm i\corr{\kappa}$. Hence $\wp_\pm$ and $(1+\mathfrak G_\pm\ast v_0)$ are invertible by Lemma \ref{L:eigen-green}. Together with \eqref{E:1-ode}, Lemma \ref{L:green-local}, and $\frac {\varphi_+(x,\lambda)}{a(\lambda)}=\vartheta_-(x,\lambda)$,  
\beq\label{E:origin}
\ba{rll}
 m(x,y,\lambda) 
=&\frac 1{a(\lambda)}\left(\wp_++\omega_+\ast v_0\right)^{-1}\varphi_+ 
&{}\\
=&\frac 1{a(\lambda)}\sum_{j=0}^\infty(-1)^j\left(\wp_+^{-1}\omega_+\ast v_0\right)^{j}\wp_+^{-1}\varphi_+,&D_{+i\corr{\kappa}}^\times;\\
m(x,y,\lambda) 
=&\left(\wp_-+\omega_-\ast v_0\right)^{-1}\vartheta_- &{}\\
=&\sum_{j=0}^\infty(-1)^j\left(\wp_-^{-1}\omega_-\ast v_0\right)^{j}\wp_-^{-1}\vartheta_-, &D_{-i\corr{\kappa}}^\times.
\ea
\eeq So 
$\frac{2i\corr{\kappa}}{\lambda-i\corr{\kappa}}\wp_+^{-1}\varphi_+(x,y,+i\corr{\kappa})$, $
\wp_-^{-1}\vartheta_-(x,y,-i\corr{\kappa})$ are  leading terms  at $+ i \corr{\kappa}$, $- i \corr{\kappa}$. 
 Defining $\Theta_{\pm}(x,y)$ and $\gamma_{\pm}$ by \eqref{E:gamma}, \eqref{E:gamma-1}, and  
using
\beq\nonumber
\ba{rl}
  &\left((1+{\mathfrak G}_+\ast v_0)^{-1}[\frac {-2}\pi\mathfrak g(x,x',i\corr{\kappa})\cot^{-1}\frac{\corr{\kappa}-|\lambda_I|}{|\lambda_R|}
]\ast v_0\right) \Theta_{+}(x,y) \\
= &-\gamma_+ \cot^{-1}\frac{\corr{\kappa}-|\lambda_I|}{|\lambda_R|}\,\Theta_{+}(x,y);
\\
  &
\left((1+{\mathfrak G}_-\ast v_0)^{-1}[\frac {-2}\pi\mathfrak g(x,x',-i\corr{\kappa})\cot^{-1}\frac{1-|\lambda_I|}{|\lambda_R|}
]\ast v_0\right) \Theta_{-}(x,y) \\
= &-\gamma_- \cot^{-1}\frac{\corr{\kappa}-|\lambda_I|}{|\lambda_R|}\,\Theta_{-}(x,y), \ea
\eeq
we obtain
\beq\label{E:m-pm-asy}
\ba{c}
 \wp_+^{-1}\left[\varphi_+(x,i\corr{\kappa})\right] 
=\frac {\Theta_{+}(x,y)}{1+\gamma_+ \cot^{-1}\frac{\corr{\kappa}-|\lambda_I|}{|\lambda_R|}},\  
 \wp_-^{-1}\left[\vartheta_-(x,-i\corr{\kappa})\right]=\frac {\Theta_{-}(x,y)}{1+\gamma_- \cot^{-1}\frac{\corr{\kappa}-|\lambda_I|}{|\lambda_R|}}.
\ea
\eeq So the remainders are 
\beq\label{E:m-r-def-new}
\begin{array}{rl}
  &m_{+,r}(x,y,\lambda)\\
= &\wp_+^{-1} \varphi_+(x,i\corr{\kappa})
+ \wp_+^{-1}\frac{ \varphi_+(x,\lambda)-\varphi_+(x,i\corr{\kappa})}{a(\lambda)} + \frac {
\sum_{j=1}^\infty  {(-1)^j}\left(\wp^{-1}_+\omega_+\ast v_0\right)^j \wp_+^{-1}\varphi_+(x,\lambda)}{a(\lambda)} ;\\
&m_{-,r}(x,y,\lambda)\\
= & \sum_{j=0}^\infty  {(-1)^j}\left(\wp^{-1}_-\omega_-\ast v_0\right)^j \wp_-^{-1}\vartheta_-(x,\lambda) - \wp_-^{-1}\left[\vartheta_-(x,-i\corr{\kappa})\right] .
\end{array}
\eeq
 Along with Lemma \ref{L:green-local}, we then prove \eqref{E:define-remainder-m} - \eqref{E:m-pm-i-0-+}.  

\vskip.1in
$\underline{\emph{Step 2 (Proof for \eqref{E:m-reality-1})}}:$ 
Condition \eqref{E:m-reality-1} can be shown by the reality of $v_0(x,y)$, a change of variables $\lambda'\mapsto -\lambda'$ in \eqref{E:a-1}, and 
\beq\label{E:pm-sym-eigen-conjugate}
\ba{c}
\overline{\phi_\pm (x,-\overline\lambda)}=\phi_\pm (x, \lambda),\ 
\overline{\psi_\pm (x,-\overline\lambda)}=\psi_\pm (x, \lambda),\ 
\overline{a ( -\overline\lambda)}=a ( \lambda) 
\ea 
\eeq  to prove 
\[\ba{c}\overline{G_c(x,x',y,-\overline\lambda)}=G_c(x,x',y,\lambda),\\ \overline{G_d(x,x',y,-\overline\lambda)}=G_d(x,x',y,\lambda). \ea\]

\vskip.1in
$\underline{\emph{Step 3 (Proof for \eqref{E:gamma-pm-est} and \eqref{E:residue})}}:$ We first claim
\beq\label{E:G-sym}
\ba{c}G(x,x',y, i\corr{\kappa}+0^+e^{i\alpha})=G(x,x',y,-i\corr{\kappa}+0^+e^{i(\pi+\alpha)}).
\ea
\eeq Combining with $\varphi_+(x,i\corr{\kappa})=\vartheta_-(x,-i\corr{\kappa})e^{-2\corr{\kappa}  x}$, $\varphi_-(x,-i\corr{\kappa})=\vartheta_+(x,i\corr{\kappa})e^{-2 \corr{\kappa} x}$, \eqref{E:chi-a}, and \eqref{E:g-asym-pm-i-prep-limit}, we obtain the commutative condition
\begin{equation}\label{E:green-exp-1}
\ba{c}
 {\mathfrak G}_+(x,x',y)e^{-2\corr{\kappa}x'}=e^{-2\corr{\kappa}x}{\mathfrak G}_-(x,x',y). 
\ea
\end{equation}
 Consequently, 
\begin{equation}\label{E:green-exp-2}
\begin{array}{c}
\Theta_+(x,y)=e^{-2\corr{\kappa} x}\Theta_-(x,y),\ \
\gamma_+=\gamma_- 
\end{array}
\end{equation}which,  combining with \eqref{E:arccot}, prove  \eqref{E:gamma-pm-est} and \eqref{E:residue}.

We now exploit the approach in \cite[Proposition 9 (i)] {VA04} to prove \eqref{E:G-sym}. In view of \eqref{E:green-1-d}, \eqref{E:tilde-g-eta-new}, and,  for fixed  $x$, $y\ne 0$, $-\frac \pi 2<\alpha\le \frac{3\pi}2$, $\alpha\ne 0,\,\pi$, let 
\[
\ba{c}
\lambda_+=\lambda_{+,R}+i\lambda_{+,I}=+ i\corr{\kappa}+0^+e^{i\alpha},\\\lambda_-=\lambda_{-,R}+i\lambda_{-,I}=- i\corr{\kappa}+0^+e^{i(\pi+\alpha)}. \ea
\]Then 
\beq\label{E:G-sym-+}
\ba{rl}
& G_{\CC^+}(x,x',y,\lambda_+)   \\
= &\theta(y) \int_{\mathbb R}e^{(\lambda^2_+-\left[\eta+i\lambda_{+,I} \right]^2)y} \frac{\phi_+(x,\eta+i\lambda_{+,I})\psi_+(x',\eta+i\lambda_{+,I})}{2\pi a_+(\eta+i\lambda_{+,I})}d\eta\\
& -  \int_{-|\lambda_{+,R}|}^{|\lambda_{+,R}|} e^{(\lambda^2_+-\left[\eta+i\lambda_{+,I} \right]^2)y} \frac{\phi_+(x,\eta+i\lambda_{+,I})\psi_+(x',\eta+i\lambda_{+,I})}{2\pi a_+(\eta+i\lambda_{+,I})}d\eta;\\
& G_{\CC^-}(x,x',y,\lambda_-)   \\
= &\theta(y) \int_{\mathbb R}e^{(\lambda^2_--\left[\eta+i\lambda_{-,I} \right]^2)y} \frac{\phi_-(x',\eta+i\lambda_{-,I})\psi_-(x,\eta+i\lambda_{-,I})}{2\pi a_-(\eta+i\lambda_{-,I})}d\eta \\
& -  \int_{-|\lambda_{-,R}|}^{|\lambda_{-,R}|} e^{(\lambda^2_--\left[\eta+i\lambda_{-,I} \right]^2)y} \frac{\phi_-(x',\eta+i\lambda_{-,I})\psi_-(x,\eta+i\lambda_{-,I})}{2\pi a_-(\eta+i\lambda_{-,I})}d\eta.
\ea
\eeq Deforming the contour, applying the residue theorem and \eqref{E:pm-sym-eigen}, 
\beq\label{E:G-sym-+-1}
\ba{rl}
&\theta(y) \int_{\mathbb R}e^{(\lambda^2_+-\left[\eta+i\lambda_{+,I} \right]^2)y} \frac{\phi_+(x,\eta+i\lambda_{+,I})\psi_+(x',\eta+i\lambda_{+,I})}{2\pi a_+(\eta+i\lambda_{+,I})}d\eta\\
=&\theta(y) \int_{\mathbb R}e^{ -\eta^2 y} \frac{\phi_+(x,\eta)\psi_+(x',\eta )}{2\pi a_+(\eta)}d\eta\\
&+2\corr{\kappa}\theta(y)[1-\theta(\corr{\kappa}-|\lambda_{+,I}|)] \phi_+(x,i\corr{\kappa})\psi_+(x',i\corr{\kappa})\\
=&\theta(y) \int_{\mathbb R}e^{ -\eta^2 y} \frac{\phi_-(x',\eta)\psi_-(x,\eta )}{2\pi a_-(\eta)}d\eta\\
&+2\corr{\kappa}\theta(y)[1-\theta(\corr{\kappa}-|\lambda_{-,I}|)] \phi_-(x',-i\corr{\kappa})\psi_-(x,-i\corr{\kappa})\\
=&\theta(y) \int_{\mathbb R}e^{(\lambda^2_--\left[\eta+i\lambda_{-,I} \right]^2)y} \frac{\phi_-(x',\eta+i\lambda_{-,I})\psi_-(x,\eta+i\lambda_{-,I})}{2\pi a_-(\eta+i\lambda_{-,I})}d\eta .
\ea
\eeq
 On the other hand, the residue theorem, \eqref{E:residue-1}, \eqref{E:arccot},   \eqref{E:green-1-d}, and the dominated convergence theorem imply
\beq\label{E:G-sym-+-2}
\ba{rl}
&\int_{-|\lambda_{+,R}|}^{|\lambda_{+,R}|} e^{(\lambda^2_+-\left[\eta+i\lambda_{+,I} \right]^2)y} \frac{\phi_+(x,\eta+i\lambda_{+,I})\psi_+(x',\eta+i\lambda_{+,I})}{2\pi a_+(\eta+i\lambda_{+,I})}d\eta\\
=&\frac {i\corr{\kappa}}\pi\phi_+(x,i\corr{\kappa})\psi_+(x',i\corr{\kappa})\int_{-|\lambda_{+,R}|}^{|\lambda_{+,R}|}  \frac{1}{\eta+i(\lambda_{+,I}-\corr{\kappa})}d\eta\\
=&-\frac {i\corr{\kappa}}\pi\phi_-(x,-i\corr{\kappa})\psi_-(x',-i\corr{\kappa})\int_{-|\lambda_{-,R}|}^{|\lambda_{-,R}|}  \frac{1}{\eta+i(\lambda_{-,I}+\corr{\kappa})}d\eta\\
=&\int_{-|\lambda_{-,R}|}^{|\lambda_{-,R}|} e^{(\lambda^2_--\left[\eta+i\lambda_{-,I} \right]^2)y} \frac{\phi_-(x',\eta+i\lambda_{-,I})\psi_-(x,\eta+i\lambda_{-,I})}{2\pi a_-(\eta+i\lambda_{-,I})}d\eta.
\ea
\eeq
Consequently, $ G_{\CC^+}(x,x',y,\lambda_0)= G_{\CC^-}(x,x',y,-\lambda_0)$ and \eqref{E:G-sym}  follow  from \eqref{E:G-sym-+}-\eqref{E:G-sym-+-2}.
\end{proof}

\begin{example}\label{Ex:m-v-0}
If $v_0(x,y)\equiv 0$, then 
\beq\label{E:m-v=0}
\ba{c}
s_d\equiv 2i\kappa,\ s_c(\lambda)\equiv 0,\ \gamma_\pm=0,\\
m(x,y,\lambda)=\frac{\varphi_+(x,\lambda)}{a_+(\lambda)}=
\vartheta_-(x,\lambda),\\  
m_{+,-1}(x,y,i\corr{\kappa}+0^+e^{i\alpha})=2i\corr{\kappa}\Theta_+(x,y)\equiv   \frac {2i\corr{\kappa}}{1+e^{2\corr{\kappa}x}},\\
m_{-,0}(x,y,-i\corr{\kappa}+0^+e^{i\alpha})= \Theta_-(x,y)\equiv  \frac {1}{1+e^{-2\corr{\kappa}x}} .  
\ea
\eeq 
\end{example}

\subsection{Continuous scattering data of $m(x,y,\lambda)$  }

\begin{lemma}\label{E:dbar-green-sd}
For $\lambda_R\ne 0$,
\begin{equation}\label{E:lambda-R-n-0}
\ba{c}
\partial_{\bar\lambda}\widetilde G(x,x',y,\lambda)=
 -\frac {\mbox{sgn}(\lambda_R)}{2\pi a(-\overline \lambda)} e^{(\lambda^2-\overline\lambda^2)y+i(\lambda+\overline\lambda)(x-x')}\mathfrak g(x,x',-\overline\lambda),
\ea
\end{equation}where $a(\lambda)$, $\mathfrak g$ are defined by \eqref{E:chi-a}, \eqref{E:g-def}.
 
 \end{lemma}

\begin{proof}   For convenience, we sketch the proof of \cite[\S 5, Proposition 11]{VA04}.   
 Define
 \[
 \ba{c}
\mathfrak F(x,x',y,\lambda,\lambda')=e^{(\lambda^2-[\lambda+\lambda']^2)y+i(\lambda-[\lambda+\lambda'])(x-x')}\mathfrak g(x,x',\lambda+\lambda') ,
\ea \] where $\mathfrak g$ is defined by \eqref{E:g-def}. So for $\lambda_R\ne 0$,  
\beq\label{E:ABC}
\ba{rl}
&\partial_{\overline\lambda}\widetilde G (x,x',y,\lambda)\\
=&\int_{\mathbb R}\frac{\mathfrak F(x,x',y,\lambda,\lambda')}{ 2\pi a(\lambda+\lambda')}\frac{\partial}{\partial\overline\lambda}\left(\theta(y)\chi_{-}-\theta(-y)\chi_{+}\right)d\lambda'\\
&+\int_{\mathbb R}\left(\theta(y)\chi_{-}-\theta(-y)\chi_{+}\right)\mathfrak F(x,x',y,\lambda,\lambda')\frac{\partial}{\partial\overline\lambda}\frac{1}{2\pi a(\lambda+\lambda')}d\lambda'\\
&-\frac{\partial}{\partial\overline\lambda}\left[2\theta(-y)\theta(1-|\lambda_I|)\frac{e^{(\lambda^2+1)y+i\lambda(x-x')}}{(e^{-x}+e^x)(e^{-x'}+e^{x'})}\right]\\
=&A+B+C.
\ea
\eeq Using
\[
\begin{array}{c}
\partial_{\overline\lambda}\chi_{\pm}=\frac 12\partial_{\lambda_R}\chi_{\pm}=\pm \texttt{sgn}(\lambda_R)\delta_{\lambda'=-2\lambda_R},\\
\frac 1{\pi}\partial_{\bar \lambda}\left(\frac 1{\lambda-a}\right)=\delta_{\lambda_R=a_R}\delta_{\lambda_I=a_I},\ 
\partial_{\overline\lambda}\theta(1\mp\lambda_I)=\frac i2\partial_{\lambda_I}\theta(1\mp\lambda_I)=\mp\frac i2 \delta_{\lambda_I=\pm 1},                                                                                                                                                                                                                                                                                                                                                                                                                                                                                                                                                                                                                                                                                                                                                                                                                                                                                                                                                                                                                                                 \end{array}
\]
one has
\beq
\ba{rl}\label{E:C}
A 
=&\int_{\mathbb R}\frac{\mathfrak F(x,x',y,\lambda,\lambda')}{ 2\pi a(\lambda+\lambda')}\frac{\partial}{\partial\overline\lambda}\left(\theta(y)\chi_{-}-\theta(-y)\chi_{+}\right)d\lambda'\\
=&\int_{\mathbb R}\frac{\mathfrak F(x,x',y,\lambda,\lambda')}{2\pi a(\lambda+\lambda')}\times(-\theta(y)\mbox{sgn}(\lambda_R)\delta_{\lambda'=-2\lambda_R}\\
&-\theta(-y)\mbox{sgn}(\lambda_R)\delta_{\lambda'=-2\lambda_R})d\lambda'\\
=&-\mbox{sgn}(\lambda_R)\frac{\mathfrak F(x,x',y,\lambda,-2\lambda_R)}{2\pi a( -\overline\lambda)}, \\
B
=&\int_{\mathbb R}\left(\theta(y)\chi_{-}-\theta(-y)\chi_{+}\right)\mathfrak F(x,x',y,\lambda,\lambda')\frac{\partial}{\partial\overline\lambda}\frac{1}{ 2\pi a(\lambda+\lambda')}d\lambda'\\
=&\int_{\mathbb R}\left(\theta(y)\chi_{-}-\theta(-y)\chi_{+}\right)\mathfrak F(x,x',y,\lambda,\lambda')(\pm i)\\
&\times\delta_{\lambda'=-\lambda_R}\delta_{\lambda_I=\pm 1}d\lambda'\\
=&\mp i\theta(-y)\mathfrak F(x,x',y,\lambda_R\pm i,-\lambda_R)\delta_{\lambda_I=\pm 1} \\
=& \mp i\theta(-y)\frac{e^{(\lambda^2+1)y+i\lambda(x-x')}}{(e^{-x}+e^x)(e^{-x'}+e^{x'})}\delta_{\lambda_I=\pm 1}, 

\ea
\eeq
\[
\ba{rl}
C
=&\pm   i\theta(-y)\frac{e^{(\lambda^2+1)y+i\lambda(x-x')}}{(e^{-x}+e^x)(e^{-x'}+e^{x'})}\delta_{\lambda_I=\pm 1}.
\ea
\]
Plugging  \eqref{E:C} into \eqref{E:ABC}, we prove the lemma.
\end{proof}

\begin{theorem}\label{T:sd-continuous}   
If $\partial_y^j\partial_x^kv_0\in  {L^1\cap L^\infty}$, $ 0\le j,\, k\le 2$, $ {|v_0|_{L^1\cap L^\infty}}\ll 1$, then
\beq\label{E:dbar-m}  
\begin{gathered}
\partial_{\overline\lambda}m(x,y,\lambda)
=    s_c(\lambda)e^{i(4\lambda_R\lambda_Iy+2\lambda_Rx)}m(x,y,-\overline\lambda),\ \lambda_R\ne 0,
\end{gathered}  
\eeq  (cf. \cite[\S 6, Eq.(46), (51)]{VA04}) where $\mathfrak s_c(\lambda)$ is defined by
\begin{equation}\label{E:alpha-def}
\begin{array}{c}
  s_c(\lambda)=\left\{ 
{\begin{array}{ll}
{\frac {\mbox{sgn}(\lambda_R)}{2\pi}}\widehat{\vartheta_+ v_0m}(\frac {\lambda_R}{\pi} ,\frac{2\lambda_R\lambda_I}\pi;\lambda),&\lambda_R\ne 0,\,\lambda\in\mathbb C^+,\\
{\frac {\mbox{sgn}(\lambda_R)}{2\pi a(-\overline \lambda)}}\widehat{\varphi_- v_0m}(\frac {\lambda_R}\pi,\frac{2\lambda_R\lambda_I}\pi;\lambda),&\lambda_R\ne 0,\,\lambda\in\mathbb C^-,
\end{array}}
\right.\\
\widehat{\vartheta_+ v_0m}(\frac {\lambda_R}\pi,\frac{2\lambda_R\lambda_I}\pi;\lambda)=
\iint e^{-i(4\lambda_R\lambda_Iy+2\lambda_Rx)} {\vartheta_+(x,-\overline\lambda)}v_0m(x,y,\lambda)dxdy,
\\
\widehat{\varphi_- v_0m}(\frac {\lambda_R}\pi,\frac{2\lambda_R\lambda_I}\pi;\lambda)=
\iint e^{-i(4\lambda_R\lambda_Iy+2\lambda_Rx)} {\varphi_-(x,-\overline\lambda)}v_0m(x,y,\lambda)dxdy,
\end{array}
\end{equation}  
and $\varphi_\pm$, $\vartheta_\pm$, $a(\lambda)$  are defined by \eqref{E:eigen-x},   \eqref{E:chi-a}, and \eqref{E:chi-a}.  
\end{theorem}
\begin{proof}   
  Denote $\widetilde G(x,x',y,\lambda)$ as $\widetilde G_\lambda$ and $\rho(x, y,\lambda,-\overline\lambda)= e^{ (\lambda^2-\overline\lambda^2)y+i (\lambda+\overline\lambda)x}$. Note $\rho(x, y,\lambda,-\overline\lambda)$ is annihilated by the heat operator $p_\lambda(D)\equiv \partial_y-\partial_x^2+2i\lambda\partial_x$. So
  \[\ba{rl}
 &p_\lambda(D)f\\
 =&p_\lambda(D)e^{(\lambda^2-\overline\lambda^2)y+i(\lambda+\overline\lambda)x}e^{-(\lambda^2-\overline\lambda^2)y-i(\lambda+\overline\lambda)x}f\\
=&e^{(\lambda^2-\overline\lambda^2)y+i(\lambda+\overline\lambda)x}p_\lambda(D)e^{-(\lambda^2-\overline\lambda^2)y-i(\lambda+\overline\lambda)x}f\\
&+e^{(\lambda^2-\overline\lambda^2)y+i(\lambda+\overline\lambda)x}(-2i(\lambda+\overline\lambda)\partial_x) e^{-(\lambda^2-\overline\lambda^2)y-i(\lambda+\overline\lambda)x}f\\
=& e^{(\lambda^2-\overline\lambda^2)y+i(\lambda+\overline\lambda)x}p_{-\overline \lambda}(D)e^{-(\lambda^2-\overline\lambda^2)y-i(\lambda+\overline\lambda)x}f  
\ea\] that is, 
\begin{equation}\label{E:green-exp}
\ba{c}
\widetilde G_\lambda\ \rho(x, y,\lambda,-\overline\lambda)=\rho(x, y,\lambda,-\overline\lambda)\ \widetilde G_{-\overline\lambda}.
\ea
\end{equation}
Therefore, for  $\lambda_R\ne 0$, denoting $
\mathfrak e(x,x',y,\lambda,-\overline\lambda)=e^{(\lambda^2-\overline\lambda^2)y+i(\lambda+\overline\lambda)(x-x')}$, (\ref{E:1-ode}), \eqref{E:lambda-R-n-0}, \eqref{E:alpha-def}, and \eqref{E:green-exp},
\beq\nonumber
\ba{rl}
&\partial_{\overline\lambda}m {(x,y,\lambda)}\\
=&-(1+\widetilde G_\lambda\ast v_0)^{-1}\left(\partial_{\bar\lambda}\widetilde G_\lambda\ast v_0\right)m {(x,y,\lambda)} \\
=& 
\left\{ 
{\begin{array}{l}
(1+\widetilde G_\lambda\ast v_0)^{-1} \frac {\texttt{sgn}(\lambda_R)\mathfrak e(x,x',y,\lambda,-\overline\lambda)\varphi_+(x,-\overline\lambda)\vartheta_+(x',-\overline\lambda)}{2\pi a(-\overline \lambda)} \ast v_0m,\, 
\lambda\in\CC^+,\\
(1+\widetilde G_\lambda\ast v_0)^{-1}\frac {\texttt{sgn}(\lambda_R)\mathfrak e(x,x',y,\lambda,-\overline\lambda){\varphi_-(x',-\overline\lambda)\vartheta_-(x,-\overline\lambda)}}{{2\pi}a(-\overline \lambda)} ]\ast v_0m,\, 
\lambda\in\CC^-
\end{array}}
\right. \\
=&  s_c(\lambda)(1+\widetilde G_\lambda\ast v_0)^{-1} e^{(\lambda^2-\overline\lambda^2)y+i(\lambda+\overline\lambda)x}m(x,-\overline\lambda) \\
=&  s_c(\lambda) e^{(\lambda^2-\overline\lambda^2)y+i(\lambda+\overline\lambda)x}(1+\widetilde G_{-\overline\lambda}\ast v_0)^{-1}m(x,-\overline\lambda) \\
=&  s_c(\lambda) e^{(\lambda^2-\overline\lambda^2)y+i(\lambda+\overline\lambda)x}m{(x,y,-\overline\lambda)}.
\ea
\eeq

\end{proof}

\begin{lemma}\label{L:constraint-continuous-discrete}
Suppose $
(1+|x|+|y|)^2\partial_x^j\partial_y^k v_0\in {L^1}\cap L^\infty$, $0\le j,\,k\le 2 $, $
|v_0|_{L^1\cap L^\infty}\ll 1$, and $v_0(x,y)\in\RR$.  
  Then    
\beq
\ba{c}
| (1-E_{+i\corr{\kappa},1/2}(\lambda)-E_{-i\corr{\kappa},1/2}(\lambda) )  s_c(\lambda) |_{   L^2(|\lambda_R| d\overline\lambda \wedge d\lambda)\cap L^\infty}  \le  C|\partial_x^2\partial_y^2v_0|_{ L^1},
\ea\label{E:pm-i-new}  
\eeq
and
\begin{gather}
\ba{c} 
  s_c(\lambda)=
\left\{
{\ba{ll}
 +\frac {i\corr{\kappa}}2\frac{\texttt{sgn}(\lambda_R)}{ \lambda-i\corr{\kappa}}\mathfrak r( \lambda)+\mbox{sgn}(\lambda_R)h^+(\lambda),&\lambda\in D^ \times_{+i\corr{\kappa}},\\
 \corr{-\frac {i\corr{\kappa}}2\frac{\texttt{sgn}(\lambda_R)}{-\overline \lambda+i\corr{\kappa}} }\mathfrak r( \lambda)+\mbox{sgn}(\lambda_R) {h^-}(\lambda),&\lambda\in  D^ \times_{-i\corr{\kappa}},
\ea}
\right.
\ea\label{E:cd-decomposition} 
\end{gather}
where 
\beq
\ba{c}  
\mathfrak r(\lambda) =\frac {\gamma_+}{1+\gamma_+\cot^{-1}\frac{\corr{\kappa}-|\lambda_I|}{|\lambda_R|}}, 
\ea\label{E:kappa}
\eeq and  $E_{z,a}$, $D_z^\times$, $\cot^{-1}\frac{\corr{\kappa}-|\lambda_I|}{|\lambda_R|}$ are defined by  Definition \ref{D:terminology}, $\gamma_+$, \eqref{E:gamma}, \eqref{E:arccot}. Moreover,
\beq \label{E:cd-decomposition-new}
\ba{c}
 |\mathfrak r|_{L^\infty}\le |v_0|_{L^1}, \ 
  {  { \sum_{j=0,1}|\partial_s^j h^\pm|_{L^\infty}\le C |(1+|x|+|y|)^2    v_0|_{L^1}, }  } 
 \\
 \mathfrak r(\lambda)=  {\mathfrak r(-\overline\lambda)}\in\RR,\ \ h^\pm(\lambda)=\overline{h^\pm(-\overline\lambda)} .\ea 
\eeq  

\end{lemma}
\begin{proof} $\underline{\emph{Step 1 (Proof for \eqref{E:pm-i-new})}}:$  We restrict to $\lambda\in\mathbb C^+$ since proofs for $\lambda\in\mathbb C^\pm$ are identical. From  \eqref{E:alpha-def}, the Fourier theory, and Theorem \ref{T:KP-eigen-existence},  
\beq\label{E:l2-linfty}
\ba{rl}
& {|(1-E_{+i\corr{\kappa},1/2}(\lambda)-E_{-i\corr{\kappa},1/2}(\lambda) )  s_c(\lambda)|}\\
\le & C|(1-E_{+i\corr{\kappa},1/2}(\lambda)-E_{-i\corr{\kappa},1/2}(\lambda) )\widehat{\vartheta_+v_0m}(\frac{\lambda_R}\pi,\frac{2\lambda_R\lambda_I}\pi,\lambda)|\\
\le & C\frac {|\partial_x^{\gamma_1}\partial_y^{\gamma_2}v_0|_{L^1}}{1+|\lambda_R|^{\gamma_1}+|\lambda_R\lambda_I|^{\gamma_2}}.
\ea
\eeq 
  Therefore \eqref{E:pm-i-new} follows from 
 \beq\label{E:high-der-lambda}
 \ba{rl}
 &\iint{|(1-E_{+i\corr{\kappa},1/2}(\lambda)-E_{-i\corr{\kappa},1/2}(\lambda) )  s_c(\lambda)|^2}|\lambda_R|d\overline\lambda\wedge d\lambda\\
 \le\ &C\iint\frac{|\partial_x^2\partial_y^2v_0|^2_{L^1}}{(1+|\lambda_R|^2+|\lambda_R\lambda_I|^2)^2}|\lambda_R|d\overline\lambda\wedge d\lambda\le \ C |\partial_x^2\partial_y^2v_0|^2_{L^1}.
  \ea
\eeq 

$\underline{\emph{Step 2 (Proof for \eqref{E:kappa} - \eqref{E:cd-decomposition-new})}}:$ For $\lambda\in D_{\pm i\corr{\kappa}}^\times$, $\lambda_R\ne 0$, from \eqref{E:alpha-def}, Theorem \ref{T:m-lambda},
\[
\ba{ rl}
   &s_c(\lambda)\\
=  &\left\{ 
{\begin{array}{l}
\frac {\mbox{sgn}(\lambda_R)}{2\pi (\lambda-i\corr{\kappa})}\iint \frac {e^{-i(4\lambda_R\lambda_Iy+2\lambda_Rx)}
\vartheta_+ (x,\lambda)v_0(x,y)(2i\corr{\kappa})\Theta_+(x,y)}{1+\gamma_+\cot^{-1}\frac{\corr{\kappa}-|\lambda_I|}{|\lambda_R|}}dxdy+ \frac {\mbox{sgn}(\lambda_R)}{2\pi }\\
\times\iint  {e^{-i(4\lambda_R\lambda_Iy+2\lambda_Rx)}
\vartheta_+ (x,\lambda)v_0(x,y) m_{+,r}(x,y,\lambda)} dxdy,\ \lambda\in\mathbb C^+,\\
\frac {\mbox{sgn}(\lambda_R)}{2\pi \frac{-\overline\lambda+i\corr{\kappa}}{-\overline\lambda-i\corr{\kappa}}}\iint \frac {e^{-i(4\lambda_R\lambda_Iy+2\lambda_Rx)}
\varphi_- (x,\lambda)v_0(x,y) \Theta_-(x,y)}{1+\gamma_-\cot^{-1}\frac{\corr{\kappa}-|\lambda_I|}{|\lambda_R|}}dxdy+ \frac {\mbox{sgn}(\lambda_R)}{2\pi\frac{-\overline\lambda+i\corr{\kappa}}{-\overline\lambda-i\corr{\kappa}} }\\
\times\iint  {e^{-i(4\lambda_R\lambda_Iy+2\lambda_Rx)}
\varphi_-(x,\lambda)v_0(x,y) m_{-,r}(x,y,\lambda)} dxdy,\ \lambda\in\mathbb C^-,
\end{array}}
\right.  

\end{array}
\]
\[
\ba{rl}
=&  \left\{ 
{\begin{array}{ll}
+\frac {i\corr{\kappa}} 2\frac {\mbox{sgn}(\lambda_R)}{\lambda-i\corr{\kappa}} \mathfrak r(\lambda)+\mbox{sgn}(\lambda_R)h^+(\lambda),& \lambda\in\mathbb C^+,\\
-\frac {i\corr{\kappa}} 2\frac {\mbox{sgn}(\lambda_R)}{-\overline\lambda+i\corr{\kappa}} \mathfrak r(\lambda)+\mbox{sgn}(\lambda_R)h^-(\lambda).&\lambda\in\mathbb C^-,
\end{array}}
\right.

\ea
\] where   
\beq\label{E:def-f-r--} 
\begin{array}{rl}
& h^+(\lambda) \\
=& \frac {1}{2\pi}\iint \frac {
\frac{\vartheta_+ (x, \lambda)-\vartheta_+ (x,i\corr{\kappa})}{\lambda-i\corr{\kappa}}v_0(x,y)(2i\corr{\kappa})\Theta_+(x,y)}{1+\gamma_+\cot^{-1}\frac{\corr{\kappa}-|\lambda_I|}{|\lambda_R|}}dxdy\\
&+ \frac {1}{2\pi }\iint  {e^{-i(4\lambda_R\lambda_Iy+2\lambda_Rx)}
\vartheta_+ (x, \lambda)v_0(x,y) m_{+,r}(x,y,\lambda)} dxdy   \\
&+  \frac {1}{2\pi}\iint \frac {
\frac{e^{-i(4\lambda_R\lambda_Iy+2\lambda_Rx)} -1}{\lambda-i\corr{\kappa}}}{1+\gamma_+\cot^{-1}\frac{\corr{\kappa}-|\lambda_I|}{|\lambda_R|}} \vartheta_+ (x,\lambda)v_0(x,y)(2i\corr{\kappa})\Theta_+(x,y)dxdy ,\\
&h^-(\lambda)\\ 
=& \frac {1}{2\pi}\iint \frac {
\frac{-\overline\lambda+i\corr{\kappa} }{-\overline\lambda+i\corr{\kappa}}\varphi_- (x,-i\corr{\kappa})v_0(x,y)\Theta_-(x,y)}{1+\gamma_-\cot^{-1}\frac{\corr{\kappa}-|\lambda_I|}{|\lambda_R|}}dxdy\\
&+ \frac {1}{2\pi}\iint \frac {
\frac{\varphi_-(x, \lambda)-\varphi_- (x,-i\corr{\kappa})}{-\overline\lambda+i\corr{\kappa}}v_0(x,y)(-\overline\lambda-i\corr{\kappa})\Theta_-(x,y)}{1+\gamma_-\cot^{-1}\frac{\corr{\kappa}-|\lambda_I|}{|\lambda_R|}}dxdy \\
&+ \frac {1}{2\pi }\frac{-\overline\lambda-i\corr{\kappa}}{-\overline\lambda+i\corr{\kappa}}\iint  {e^{-i(4\lambda_R\lambda_Iy+2\lambda_Rx)}
\varphi_- (x, \lambda)v_0(x,y) m_{-,r}(x,y,\lambda)} dxdy   \\
&+  \frac {1}{2\pi}\iint \frac {
\frac{e^{-i(4\lambda_R\lambda_Iy+2\lambda_Rx)} -1}{-\overline\lambda+i\corr{\kappa}}}{1+\gamma_-\cot^{-1}\frac{\corr{\kappa}-|\lambda_I|}{|\lambda_R|}} \varphi_- (x,\lambda)v_0(x,y)(-\overline\lambda-i\corr{\kappa})\Theta_-(x,y)dxdy    .
\ea\eeq  
Combining with the reality of $\cot^{-1}\frac{\corr{\kappa}-|\lambda_I|}{|\lambda_R|}$, \eqref{E:gamma-pm-est}, we justify 
 \eqref{E:kappa} - \eqref{E:cd-decomposition-new}.

\end{proof}

\subsection{The  forward scattering
 transform and the eigenfunction space}  

Due to the discontinuities of  $m_{+,-1}$, the kernel of the Cauchy equation ($\overline{\partial}$-equation)  for $m$ is not 
\[
\ba{c}
\partial_{\overline\lambda}m(x,y,\lambda)=   s_c(\lambda)e^{i(4\lambda_R\lambda_Iy+2\lambda_Rx)}m(x,y,-\overline\lambda)
\ea
\]  simply. We need to correct it by subtracting the contribution from  $m_{+,-1}$. 

Another neater alternative is through a regularization to derive the Cauchy equation for the inverse problem. 
From \eqref{E:l2-linfty},  convergence of the  Cauchy integral equation \cite[Eq.(5.10)]{BP214}  can not be achieved for generic initial data $v_0(x,y)$. Nevertheless, \cite{BP214}    
made an important progress in dealing the singular integrals at $\kappa_n$ for the perturbed KPII multi-line solitons equations.  Precisely,
\begin{itemize}
\item the singularities at $\kappa_n$ of the Cauchy integral equation  are regularized by renormalizing $m(x,y,\lambda)$ by holomorphic functions vanishing at $\kappa_n$ \cite[Eq.(2.12)]{BP214};
\item Cauchy integrals of the leading singularities of  the regularized eigenfunction  at $\kappa_n$  are integrated by applying Stokes' theorem, and the reality property  at $\kappa_n$  is proved \cite[Lemma 5.1]{BP214};
\item boundary terms of their Cauchy integral equation \cite[Eq.(5.33)]{BP214} are characterized by the asymptotic properties   at $\kappa_n$ (however, base on the convergence of \cite[Eq.(5.10)]{BP214}).
\end{itemize}

Inspired by these results, we introduce the regularized eigenfunction 
\beq\label{E:define-remainder-m-sf-0}
\ba{c}
\mathfrak m (x,y,\lambda)=\frac {\lambda-i\corr{\kappa}}{\lambda-2i\corr{\kappa}}m(x,y,\lambda),
\ea
\eeq   to tame the singularities at  $+i\corr{\kappa},\, \infty$ and keep the  symmetry. 

\begin{theorem}\label{T:normal-eigen} Suppose 
\[
\ba{c}
u_0(x )=-2\corr{\kappa^2}\textrm{sech}^2\corr{\kappa}x,\, \kappa>0,\\
(1+|x|+|y|)^2\partial_x^j\partial_y^k v_0\in {L^1}\cap L^\infty,\ 0\le j,\,k\le 2,\\
|v_0|_{L^1\cap L^\infty}\ll 1,\ v_0(x,y)\in\RR.
\ea
\]  Then for fixed $\lambda\in\CC\backslash \{\pm i\corr{\kappa},\,
\iota\}$, there exists a unique   eigenfunction $\mathfrak m$ satisfying
\beq\label{E:normal-ef-pde}
\ba{c}
\left(\partial_y-\partial_x^2+2i\lambda\partial_x+u_0(x)\right){\mathfrak m}=-v_0(x,y){\mathfrak m},\\
\lim_{(x,y)\to\infty}  ({\mathfrak m}-\frac {\lambda-i\corr{\kappa}}{\lambda-2i\corr{\kappa}}\vartheta_-(x,\lambda)  )=0,
\ea
\eeq and    
\begin{gather}
  \ba{c}
  |(1- E_{\iota} )\mathfrak m(x,y,\lambda)|\le C;\\
  \\
  {\mathfrak m}(x,y,\lambda)  
 =    \frac{\mathfrak m_{
res}(x,y)}{\lambda-\iota}  +\mathfrak m_{\iota,r}(x,y,\lambda),
 \quad\lambda\in D^ \times_{\iota},\\
\mathfrak m_{res}(x,y)\in i\RR, \ \
|\mathfrak m_{res}(x,y)|_{L^\infty}\le C, \\
|(\lambda-\iota)\mathfrak m_{\iota,r}(x,y,\lambda) |_{L^\infty}\le  C,\ |\mathfrak m_{\iota,r}(x,y,\lambda) |_{L^\infty}\le  C(1+|x|+|y|) ;
\ea \label{E:define-remainder-m-sf-pm2} 
\end{gather}
\begin{gather}
 \ba{c}{\mathfrak m}(x,y,\lambda)  
 = 
 \mathfrak m_{\pm i\corr{\kappa},0}(x,y,\lambda)+\mathfrak m_{\pm i\corr{\kappa} ,r}(x,y,\lambda),
\quad\lambda\in D^\times_{\pm i\corr{\kappa}},\\ 
\mathfrak m_{+i\corr{\kappa},0}(x,y,\lambda)= {\frac{   -2\Theta_+(x,y)}{1+\gamma_+\cot^{-1}\frac {\corr{\kappa}-|\lambda_I|}{|\lambda_R|}}},
\ 
\mathfrak m_{-i\corr{\kappa},0}(x,y,\lambda)= {\frac{  \frac 23 \Theta_-(x,y)}{1+\gamma_+\cot^{-1}\frac {\corr{\kappa}-|\lambda_I|}{|\lambda_R|}}} , \\
\mathfrak m_{+i\corr{\kappa},0}(x,y,+i\corr{\kappa}+0^+e^{i\alpha})=  {\mathfrak s}_de^{-2\corr{\kappa}x}\mathfrak m_{-i\corr{\kappa},0}(x,y,-i\corr{\kappa}+0^+e^{i(\pi+\alpha)}) ,   \\
\mathfrak m_{ \pm i \corr{\kappa},r}(x,y,\pm i\corr{\kappa})=0,\ |\mathfrak m_{ \pm i\corr{\kappa},r}(x,y,\lambda) |_{L^\infty}\le C,\\
|\frac\partial{\partial s} \mathfrak m_{ \pm i\corr{\kappa},r}(x,y,\lambda) |_{L^\infty}\le C(1+|x|+|y|), 
  \ea  \label{E:define-remainder-m-sf-new}
  \end{gather}
where $\vartheta_-(x,\lambda)$ is defined by \eqref{E:eigen-x}, $\lambda=z+se^{i\alpha}$, $z\in\mathfrak Z=\{\pm i\kappa,\iota\}$,  $E_z$,  $D_z^\times$, $\iota$,  are defined by Definition \ref{D:terminology}, $\Theta_\pm$, $ \corr{\mathfrak s_d=-3}$, $\gamma_+$, and $\cot^{-1}\frac{\corr{\kappa}-|\lambda_I|}{|\lambda_R|}$  defined by \eqref{E:gamma}, \eqref{E:gamma-1}, \eqref{E:residue}, \eqref{E:arccot}.  

Moreover,  for $\lambda_R\ne 0$, $\lambda\ne \iota $,
\beq\label{E:dbar-m-sf}
\ba{c}
  \partial_{\overline\lambda}{\mathfrak m}(x,y,\lambda) 
=   {\mathfrak s}_c(\lambda)e^{i(4\lambda_R\lambda_Iy+2\lambda_Rx)}\mathfrak m(x,y,-\overline\lambda),\\
 {\mathfrak s}_c(\lambda)=\frac {(\lambda-i\corr{\kappa})( -\overline\lambda-\iota)}{(\lambda-\iota )( -\overline\lambda-i\corr{\kappa}) } 
   s_c(\lambda),
\ea
\eeq with
\begin{gather}
\ba{c}
{|{(1-E_{+i\corr{\kappa},1/2}(\lambda)-E_{-i\corr{\kappa},1/2} ) {\mathfrak s}_c(\lambda)}|_{   L^2(|\lambda_R| d\overline\lambda \wedge d\lambda)\cap L^\infty} }\le  C|\partial_x^2\partial_y^2v_0|_{ L^1} ;
\ea\label{E:tilde-dbar-m-sf-new-0}\\
\ba{c}
 {\mathfrak s}_c(\lambda)=   \frac i2  {\mbox{sgn}(\lambda_R)}\mathfrak c +\mbox{sgn}(\lambda_R) {\hbar}  _{\iota }(\lambda)  ,\quad\lambda\in D^\times_{\iota },\\
\sum _{0\le k\le 1}|\partial_s^k\hbar_{\iota}(\lambda)  |_{L^\infty}\le C|(1+|x|+|y|)^2   v_0|_{L^1},\\
\textit{$\mathfrak c$ constants, $\hbar_{\iota}(\iota)=0$;}
\ea\label{E:dbar-m-sf-new-1}\\
\ba{c}
 {\mathfrak s}_c(\lambda)=  \pm\frac {i\corr{\kappa}}2 \frac{\mbox{sgn}(\lambda_R)}{-\overline\lambda\mp i\corr{\kappa}}\mathfrak r( \lambda)+\mbox{sgn}(\lambda_R) {\hbar}_{\pm i\corr{\kappa}}(\lambda)  ,\quad\lambda\in D^\times_{\pm i\kappa},\\
\sum _{0\le k\le 1}| \partial_{s}^k  \hbar _{\pm i\corr{\kappa}} |_{L^\infty}\le C|(1+|x|+|y|)^2   v_0|_{L^1}  ,
\ea\label{E:dbar-m-sf-new}
\end{gather} 
 where $  s_c(\lambda)$, $\mathfrak r(\lambda)$ are defined by   \eqref{E:alpha-def}, \eqref{E:kappa}.  Finally, 
\beq\label{E:renormalize-sym}
\ba{c}
{\mathfrak m}(x,y,\lambda)=\overline{{\mathfrak m}(x,y,-\overline\lambda)},\ \
 {\hbar_z }(\lambda)=\overline{ {\hbar_z }( -\overline\lambda)},\ \ z\in \mathfrak Z=\{\pm i\corr{\kappa},\iota\}. 
\ea
\eeq

\end{theorem}

\begin{proof} Properties \eqref{E:define-remainder-m-sf-pm2}  can be derived by  the formula $\widetilde G_c$, $\widetilde G_d$, and the integral equation of $m$, namely, \eqref{E:green-1-d}, \eqref{E:chi-a}, \eqref{E:1-ode}, and \eqref{E:define-remainder-m-sf-0}. The others follow from Theorem \ref{T:KP-eigen-existence}, \ref{T:m-lambda}, \ref{T:sd-continuous}, and Lemma \ref{L:constraint-continuous-discrete}, in particular, \eqref{E:define-remainder-m}  - \eqref{E:m-pm-i-0-+},   \eqref{E:dbar-m}, \eqref{E:pm-i-new}, and \eqref{E:cd-decomposition}.

\end{proof}

\begin{example}\label{Ex:mathfrak-m-v-0}
If $v_0(x,y)\equiv 0$, then
\[
\ba{c}
\mathfrak s_d\equiv -3,\ \mathfrak s_c(\lambda)\equiv 0,\ \gamma_\pm=0,\\
 \mathfrak m(x,y,\lambda)= {\varphi_+(x,\lambda)}\frac {(\lambda+i\corr{\kappa}) }{(\lambda-2i\corr{\kappa}) }=
1+\frac{3i\kappa}{\lambda-2i\kappa}(1-\frac{\frac 23}{1+e^{-2\kappa x}}),\\ 
\mathfrak  m_{+i\corr{\kappa},0}(x,y,i\corr{\kappa}+0^+e^{i\alpha})\equiv  \frac {-2}{1+e^{2\corr{\kappa}x}},\ \  
\mathfrak m_{-i\corr{\kappa},0}(x,y,-i\corr{\kappa}+0^+e^{i\alpha})\equiv  \frac { \frac 23}{1+e^{-2\corr{\kappa}x}} ,\\
\mathfrak  m_{ res}(x,y)\equiv 3i\corr{\kappa}  (1-\frac {\frac 23}{1+e^{-2\corr{\kappa}x}}).
\ea
\]
\end{example}

Based Theorem \ref{T:normal-eigen}, we introduce the space of eigenfunctions  ${ W}$ and the spectral operator $T$ as follows.

{\begin{definition}\label{D:quadrature-hat}
  The  eigenfunction space ${ W }\equiv { W}_{x,y}$ is the set of functions 
 \[
\begin{array}{rl}
i. & \phi (x,y,\lambda)=\overline{ \phi (x,y,-\overline\lambda)};\\
ii. & (1-  E_{\iota }   )\phi(x,y,\lambda)\in L^\infty;\\
iii. & \phi(x,y,\lambda)=\frac{\phi_{res}(x,y)}{\lambda-\iota}  +\phi_{\iota,r}(x,y,\lambda),\ \lambda \in D_{\iota}^\times,
\\
&\phi_{res}(x,y), \ (\lambda-\iota)\phi_{\iota,r}(x,y,\lambda),\ 
   \frac{\phi_{\iota,r}(x,y,\lambda)}{1+|x|+|y|} \in L^\infty ;
\ea\]
\[
\ba{rl}
iv. & \phi (x,y,  \lambda)=\phi_{\pm i\kappa,0}(x,y,\lambda)+\phi_{\pm i\kappa,r}(x,y,\lambda),\ \lambda \in D_{\pm i\kappa}^\times,\\
  & \phi_{\pm i\kappa,0}(x,y,\lambda)=  
  \left\{
\ba{l}
 {\frac{  \mathfrak s_de^{-2\kappa x}a(x,y)}{1+\gamma_+\cot^{-1}\frac {\kappa-|\lambda_I|}{|\lambda_R|}}}, \\
{\frac{  a(x,y)}{1+\gamma_+\cot^{-1}\frac {\kappa-|\lambda_I|}{|\lambda_R|}}},
 \ea 
\right.  \\
& \phi_{\pm i\kappa,r}(x,y,\pm i\kappa)=0,\ \phi_{ \pm i\kappa,r} (x,y,\lambda),\  \frac{\frac\partial{\partial s} \phi_{ \pm i\corr{\kappa},r}(x,y,\lambda) }{1+|x|+|y|} \in L^\infty.
\end{array}
\]
\end{definition} 

\begin{definition}\label{D:spectral} 
Define $\{\iota, {\mathfrak s}_d,   {\mathfrak s}_c(\lambda)\}$ as the set of scattering data, 
where the pole $\iota=+2i\kappa$ and the norming constant $ {\mathfrak s}_d=-3$ are the {\bf \sl discrete scattering data};  and  $ {\mathfrak s}_c(\lambda)$, the  {\bf \sl continuous scattering data}, is defined by \eqref{E:dbar-m-sf}. Denote   $T$ as   the {\bf\sl forward scattering transform} by
\beq\label{E:cauchy-operator}
\ba{c}
T(\phi)(x,y,\lambda)  = {\mathfrak s}_c(\lambda)e^{i(4\lambda_R\lambda_Iy+2\lambda_Rx)} \phi(x,y,-\overline\lambda).
\ea
\eeq 
\end{definition}

\begin{definition}\label{D:cauchy}
Let   $\mathcal C$ be the Cauchy integral operator  defined by 
\beq\label{E:cauchy-1}
\ba{c}
\mathcal C(\phi)(x,y, \lambda)=\mathcal C_\lambda(\phi) =-\frac 1{2\pi i}\iint\frac {\phi(x,y, \zeta)}{\zeta-\lambda}d\overline\zeta\wedge d\zeta. 
\ea
\eeq Decompose
\beq\label{E:decomposition-cauchy-def}
\ba{c}
 \mathcal CT\phi=\sum_{z\in\mathfrak Z}\mathcal C  E_z T\phi+ \mathcal C \left[1- \sum_{z\in\mathfrak Z}E_z\right ]T\phi,
\ea
\eeq 
where   $E_z$ and $\mathfrak Z=\{\pm i\corr{\kappa},\iota\}$ are defined by Definition \ref{D:terminology}.  
\end{definition}
 
\subsection{The spectral analysis} 

Due to Theorem \ref{T:normal-eigen}, outside the singularities $\pm i\kappa$, $\iota$, the eigenfunction $\mathfrak m $ and  continuous scattering data ($\overline\partial$-data) $\mathfrak s_c $ possess the same  analytical properties as those for the localized KPII solutions   \cite{W85}, \cite{W87}. As a result, spectral analysis there is the same as that for vacuum background (see Lemma \ref{L:infinity-spectral} in below).   
On the other hand, from \eqref{E:dbar-m-sf-new-1}, \eqref{E:dbar-m-sf-new}, the Cauchy integrals  at $\pm i\kappa$   are two dimensional singular integrals with blowing ups of order two and highly oscillatory, not fully symmetric kernels which cause  difficulties for deriving uniform estimates for the spectral transform. On the other hand,   non-uniform estimates for   these singular integrals can be achieved by applying Stokes' or the Cauchy theorem to integrate the leading singularities ( see Lemma \ref{L:pogrebkov} and \ref{L:leading-iota}).

\begin{lemma}\label{L:pogrebkov}   For $\lambda\in D^\times_{\pm i\corr{\kappa}}$, 
\beq\label{E:pogrebkov}
\ba{c}
  -\frac {1}{2\pi i}\iint\frac { \pm\frac {i}2\mbox{sgn}(\zeta_R) \mathfrak r(\zeta) E_{\pm i\corr{\kappa}}\mathfrak m_{\pm i\corr{\kappa},0}(x,y,\zeta)}{(\zeta-\lambda)(-\overline\zeta\mp i\corr{\kappa})}d\overline\zeta\wedge d\zeta \\
=   \mathfrak m_{\pm i\corr{\kappa},0}(x,y,\lambda)-\frac {1} {2\pi i}\oint_{|\zeta\mp i\corr{\kappa}|=1}\frac{\mathfrak m_{\pm i\corr{\kappa}, 0}(x,y, \zeta)}{\zeta- \lambda}d\zeta 
\ea
\eeq \cite{BP214}, \cite{BP302}.  Here and in the following the circular integration   is taken counterclockwisely and $E_z$, $D_z^\times$, $\mathfrak r$, $\mathfrak m_{\pm i\corr{\kappa},0}$ are defined by Definition \ref{D:terminology}, \eqref{E:kappa},  \eqref{E:define-remainder-m-sf-new}. Moreover,   
\bea
&&\ba{l}
 |  \mathcal C E_{\pm i\corr{\kappa}} T\mathfrak m |_{L^\infty} 
\le  C(1+|x|+|y|),
\ea\label{E:pmi-pm2-iota-h}  
 \\ 
&&\ba{l}
\mathcal C E_{\pm i\corr{\kappa}} T\mathfrak m\to 0  \textit{ uniformly as } |\lambda|\to \infty.
\ea\label{E:pmi-pm2}
\eea  

\end{lemma}

\begin{proof} From \eqref{E:define-remainder-m-sf-new}, and for $\zeta=\pm i\corr{\kappa}+se^{i\beta}\in D_{\pm i\corr{\kappa}}^\times$,   
\beq\label{E:symbol-continuous}
\begin{array}{c}
 -\frac i 2\frac {1}{-\overline\zeta\mp i\corr{\kappa}}= \partial_{\overline\zeta}\beta,\ 
\partial_{ \overline\zeta}\cot^{-1}\frac {\corr{\kappa}-|\zeta_I|}{|\zeta_R|}=\mp\frac i2\frac {\mbox{sgn}(\zeta_R)}{-\overline\zeta\mp i\corr{\kappa}}.
\end{array}
\eeq Hence
\beq\label{E:pogrebkov-proof}
\begin{array}{c}
\partial_{ \overline\zeta}\mathfrak m_{\pm i\corr{\kappa},0}(x,y,\zeta)
=\frac { \pm\frac {i}2\mbox{sgn}(\zeta_R) \mathfrak r(\zeta) \mathfrak m_{\pm i\corr{\kappa},0}(x,y,\zeta)}{ -\overline\zeta\mp i\corr{\kappa} }.
\end{array}
\eeq    Applying Stokes' theorem,
\beq\label{E:stokes-1}
\ba{rl}
&-\frac 1{2\pi i}\iint_{D_{\pm i\kappa}/(D_{\pm i\kappa,\epsilon }\cup D_{\lambda,\epsilon})}\frac{\pm\frac {i}2\mbox{sgn}(\zeta_R) \mathfrak r(\zeta)E_{\pm i\corr{\kappa}}\mathfrak m_{\pm i\corr{\kappa},0}(x,y,\zeta)} {(\zeta-\lambda)(-\overline\zeta\mp i\corr{\kappa})}d\overline  \zeta\wedge d\zeta\\
=&-\frac 1{2\pi i}\int_{\partial(D_{\pm i\kappa,\epsilon }\cup D_{\lambda,\epsilon})} \frac{\mathfrak m_{\pm i{\kappa},0}(x,y,\zeta)}{\zeta-\lambda}d\zeta\\
=&-\frac {1} {2\pi i}\oint_{|\zeta\mp i\corr{\kappa}|=1}\frac{\mathfrak m_{\pm i\corr{\kappa}, 0}(x,y, \zeta)}{\zeta- \lambda}d\zeta+\frac 1{2\pi i}\int_{\partial D_{\pm i\kappa,\epsilon } } \frac{\mathfrak m_{\pm i{\kappa},0}(x,y,\zeta)}{\zeta-\lambda}d\zeta\\
&+\frac 1{2\pi i}\int_{\partial D_{\lambda,\epsilon} } \frac{\mathfrak m_{\pm i{\kappa},0}(x,y,\zeta)}{\zeta-\lambda}d\zeta.
\ea
\eeq Note,  by $\lambda\ne \pm i\kappa$ and Theorem \ref{T:normal-eigen},  
\beq\label{E:stokes-2}
\ba{rl}
&-\frac 1{2\pi i}\iint_{D_{\pm i\kappa,\epsilon }}\frac{\pm\frac {i}2\mbox{sgn}(\zeta_R) \mathfrak r(\zeta)E_{\pm i\corr{\kappa}}\mathfrak m_{\pm i\corr{\kappa},0}(x,y,\zeta)} {(\zeta-\lambda)(-\overline\zeta\mp i\corr{\kappa})}d\overline  \zeta\wedge d\zeta\ \to 0,\\
&-\frac 1{2\pi i}\iint_{D_{\lambda,\epsilon}}\frac{\pm\frac {i}2\mbox{sgn}(\zeta_R) \mathfrak r(\zeta)E_{\pm i\corr{\kappa}}\mathfrak m_{\pm i\corr{\kappa},0}(x,y,\zeta)} {(\zeta-\lambda)(-\overline\zeta\mp i\corr{\kappa})}d\overline  \zeta\wedge d\zeta\ \to 0,\\
&+\frac 1{2\pi i}\int_{\partial D_{\pm i\kappa,\epsilon } } \frac{\mathfrak m_{\pm i{\kappa},0}(x,y,\zeta)}{\zeta-\lambda}d\zeta\ \to 0,\\
&+\frac 1{2\pi i}\int_{\partial D_{\lambda,\epsilon} } \frac{\mathfrak m_{\pm i{\kappa},0}(x,y,\zeta)}{\zeta-\lambda}d\zeta\ \to m_{\pm i{\kappa},0}(x,y,\lambda),\ \ \textit{as $\epsilon\to 0$}.
\ea
\eeq Therefore \eqref{E:pogrebkov} follows.

Moreover, writing
\[
\ba{rl}
 &\mathcal C_\lambda  E_{\pm i\kappa} T\mathfrak m  \\
 =&-\frac 1{2\pi i}\iint\frac{  E_{\pm i\kappa}(\zeta)e^{i(4\zeta_R\zeta_Iy+2\zeta_R x)}\left[\pm\frac i2\frac{\mbox{sgn}(\zeta_R)}{-\overline\zeta\mp i\kappa}\mathfrak r(\zeta)+\mbox{sgn}(\zeta_R) \hbar_{\pm i\kappa}(\zeta)\right]}{\zeta-\lambda}
  \\
  &\times (\mathfrak m_{\pm i\kappa,0}(x,y,-\overline\zeta)+\mathfrak m_{\pm i\kappa,r}(x,y,-\overline\zeta))d\overline\zeta\wedge d\zeta\\
  =&I_1^\pm(x,y,\lambda)+I_2^\pm(x,y,\lambda),
\ea
\]where
\[
\ba{rl}
I_1^\pm(x,y,\lambda)= &\mathcal C_\lambda(\frac {\pm\frac i2\mbox{sgn}(\zeta_R)E_{\pm i\kappa}(\zeta)\mathfrak r(\zeta)\mathfrak m_{\pm i\kappa,0}(x,y,-\overline\zeta)}{-\overline\zeta\mp i\kappa} );
\ea
\]
\beq\label{E:R-decompose-formula}
\ba{rl}
 I_2^\pm(x,y,\lambda)=&\mathcal C_\lambda(\frac {\pm\frac i2\mbox{sgn}(\zeta_R)E_{\pm i\kappa}(\zeta)F_\pm(x,y,\zeta)}{-\overline\zeta\mp i\kappa}),\\
 F_\pm (x,y,\zeta)=&   [e^{i(4\zeta_R\zeta_Iy+2\zeta_R x)}-1]\mathfrak r(\zeta)\mathfrak m_{\pm i\kappa,0}(x,y,-\overline\zeta)\\
&+  e^{i(4\zeta_R\zeta_Iy+2\zeta_R x)} \mathfrak r(\lambda)\mathfrak m_{\pm i\kappa,r}(x,y,-\overline\zeta)\\
&\pm  e^{i(4\zeta_R\zeta_Iy+2\zeta_R x)}\frac 2i(-\overline\zeta\mp i)\hbar_{\pm i\kappa} (\zeta)\mathfrak m (x,y,-\overline\zeta) ,
\ea
\eeq
Using $\frac{\mathfrak m_{\pm i\kappa,r}(x,y,-\overline\zeta)}{-\overline\zeta\mp i\kappa}=\frac{\mathfrak m_{\pm i\kappa,r}(x,y,-\overline\zeta)-\mathfrak m_{\pm i\kappa,r}  (x,y,\pm i\kappa)}{-\overline\zeta\mp i\kappa}$ and via \eqref{E:define-remainder-m-sf-new}, \eqref{E:pogrebkov},  we conclude
\beq\label{E:pm-i-k-double}
\ba{c}
|E_{\pm i\corr{\kappa}}\mathcal CE_{\pm i\corr{\kappa}}T\mathfrak m|_{L^\infty}\le C(1+|x|+|y|),\\
|(1- E_{\pm i\kappa})\mathcal CE_{\pm i\corr{\kappa}}T\mathfrak m|
_{L^\infty}\le C
\ea
\eeq
which yield \eqref{E:pmi-pm2-iota-h} and \eqref{E:pmi-pm2}. 
\end{proof}

\begin{lemma}\label{L:leading-iota} Let $E_\iota$, $D_\iota$  be defined by Definition \ref{D:terminology}. Then  
\beq\label{E:pogrebkov-iota}
\ba{c}
  -\frac 1 {2\pi i}\iint\frac { \frac i2\mbox{sgn}(\zeta_R) E_{\iota} }{(\zeta-\lambda)(-\overline\zeta-\iota)}d\overline\zeta\wedge d\zeta \in L^\infty(D_{\iota}) 
\ea
\eeq which vanishes at $\iota=2i\kappa$.  Consequently,   
\bea
&&\ba{l}
 |  \mathcal C E_{\iota} T\mathfrak m |_{L^\infty} 
\le  C(1+|x|+|y|),
\ea\label{E:iota-h}  
 \\ 
&&\ba{l}
\mathcal C E_{\iota} T\mathfrak m\to 0 \textit{ uniformly as } |\lambda|\to \infty.
\ea\label{E:iota}
\eea  
\end{lemma}
\begin{proof} 
 By using the polar coordinates $\zeta= \iota +se^{i\beta}\in D_{ \iota  }^\times$,
\beq\label{E:pogrebkov-Stokes-iota}
\begin{array}{rl}
& -\frac 1 {2\pi i}\iint\frac {  \frac i2\mbox{sgn}(\zeta_R) E_{ \iota} }{(\zeta-\lambda)(-\overline\zeta-\iota)}d\overline\zeta\wedge d\zeta\\
=&
    \frac {1}{ 2\pi i }\int_{-\frac \pi 2}^{\frac{3\pi}2}d\beta \int_0^1\frac { \mbox{sgn}(\zeta_R)    }{(se^{i\beta}-(\lambda-\iota))se^{-i\beta }} sds\\
    = & \frac {1}{ 2\pi i }\int_{-\frac \pi 2}^{\frac{ \pi}2}d\beta \int_{-1}^1\frac { 1    }{s -(\lambda-\iota)e^{-i\beta}} ds,
\end{array}
\eeq 
which is a composition of the Hilbert transform. Therefore the H$\ddot{\mbox o}$lder continuity  and vanishing at $ \iota$ can be proved and \eqref{E:pogrebkov-iota} follows.  Writing
\beq\label{E:extra-factor}
\ba{rl}  
 \mathcal C  E_{\iota} T\mathfrak m  
    =&-\frac 1{2\pi i}\iint\frac{  E_{\iota}(\zeta)e^{i(4\zeta_R\zeta_Iy+2\zeta_R x)}\left[  \frac i2 \mbox{sgn}(\zeta_R)\mathfrak c+\mbox{sgn}(\zeta_R) \hbar_{\iota}(\zeta)\right]}{\zeta-\lambda}
  \\
  &\times (\frac {\mathfrak m_{res}(x,y)}{-\overline\zeta-\iota}+\mathfrak m_{\iota,r}(x,y,-\overline\zeta))d\overline\zeta\wedge d\zeta \\
  =& I _1'(x,y,\lambda)+  I _2'(x,y,\lambda),\\
 I _1'(x,y,\lambda)= &\mathcal C_\lambda(\frac { \frac i2\mbox{sgn}(\zeta_R)\mathfrak c E_{\iota}(\zeta)\mathfrak m_{res}(x,y)}{-\overline\zeta-\iota} )\\ 
I_2'(x,y,\lambda)=&\mathcal C_\lambda( \frac {\frac i2\mbox{sgn}(\zeta_R)E_{\iota}(\zeta)F (x,y,\zeta)}{-\overline\zeta-\iota}),\\
 F (x,y,\zeta)=&  [e^{i(4\zeta_R\zeta_Iy+2\zeta_R x)}-1]\mathfrak c\mathfrak m_{res}(x,y)\\
&+ e^{i(4\zeta_R\zeta_Iy+2\zeta_R x)} (-\overline\zeta-\iota)\mathfrak c\mathfrak m_{\iota,r}(x,y,-\overline\zeta)\\
&+ e^{i(4\zeta_R\zeta_Iy+2\zeta_R x)}\frac 2i(-\overline\zeta-\iota)\hbar_{\iota} (\zeta)\mathfrak m  (x,y) .
  \ea
\eeq Via \eqref{E:define-remainder-m-sf-pm2} and \eqref{E:pogrebkov-iota}, we conclude
\beq\label{E:iota-double}
\ba{c}
|E_{  \iota}\mathcal CE_{  \iota}T\mathfrak m|_{L^\infty}\le C(1+|x|+|y|),\\
|(1- E_{\iota})\mathcal CE_{  \iota}T\mathfrak m|
_{L^\infty}\le C
\ea
\eeq
which yield \eqref{E:iota-h} and \eqref{E:iota}.

\end{proof}

\begin{lemma}\label{L:infinity-spectral} Suppose $
(1+|x|+|y|)^2\partial_x^j\partial_y^k v_0\in {L^1}$, $0\le j,\,k\le 2 $, $
|v|_{L^1\cap L^\infty}\ll 1$. Let   $E_z$, $E_{z,a}$,    $ \iota$, $\mathfrak Z$ be defined by Definition \ref{D:terminology}.    
\beq
 \ba{c}
 | \mathcal C \left[1- \sum_{z\in\mathfrak Z}E_z\right ]T\mathfrak m |_{L^\infty }\le C,\ea\label{E:infinity-lambda-i-spectal} 
\eeq and
\beq\label{E:spec-infinity}
\ba{c}
\mathcal C \left[1- \sum_{z\in\mathfrak Z}E_z\right ]T\mathfrak m(x,y,\lambda)\to 0 \textit{ uniformly as $|\lambda|\to\infty$, $\lambda_R\ne 0$.}
\ea
\eeq
\end{lemma}

\begin{proof} Via a change of variables 
 \beq\label{E:wickhauser-spec}
 \ba{c}
2\pi\xi=\zeta+\overline\zeta,\quad 2\pi i\eta=\zeta^2-\overline\zeta^2,\\
\zeta=\pi\xi+i\frac \eta{2\xi}, \\ d\overline\zeta\wedge d\zeta=\frac{i\pi}{|\xi|}d\xi d\eta,
\ea
\eeq 
and from \eqref{E:pm-i-new},  \cite[Lemma 2.II]{W85}, \cite[Lemma 2.II]{W87}
\begin{equation}\label{E:wick-infty-spec}
\ba{c}
p_\lambda(\xi,\eta)=(2\pi\xi)^ 2-4\pi\xi\lambda+2\pi i \eta ,\\ \Omega_\lambda=\{(\xi,\eta)\in\RR^2\ :\ |p_\lambda(\xi,\eta)|<1\},\\
\left|\frac 1{p_\lambda}\right|_{L^1(\Omega_\lambda, d\xi d\eta)}\le  \frac C{(1+|\lambda_R|^2)^{1/2}},\quad
\left|\frac 1{p_\lambda}\right|_{L^2(\Omega_\lambda^c,d\xi d\eta)}\le \frac C{(1+|\lambda_R|^2)^{1/4}},
\ea
\end{equation}    we obtain
\beq\label{E:infty-spec}
\begin{array}{rl}
&|\mathcal C[1- \sum_{z\in\mathfrak Z}E_z ] T\mathfrak m|\\
\le & C| [1- \sum_{z\in\mathfrak Z}E_z ] \mathfrak m|_{L^\infty} \iint \frac {{|\left[1-E_{+i\corr{\kappa},1/2}(\zeta)-E_{-i\corr{\kappa},1/2}(\zeta)\right]\mathfrak s_c(\zeta)|}}{|\zeta-\lambda|}d\overline\zeta\wedge d\zeta \\
\le & C| [1- \sum_{z\in\mathfrak Z}E_z ]\mathfrak m|_{L^\infty} \iint \frac {{|(-2\pi)\left[1-E_{+i\corr{\kappa},1/2}(\zeta)-E_{-i\corr{\kappa},1/2}(\zeta)\right]\mathfrak s_c(\zeta)|}}{|(2\pi\xi)^ 2-4\pi\xi\lambda+2\pi i \eta|}d\xi d\eta\\
\le &C| [1- \sum_{z\in\mathfrak Z}E_z]\mathfrak m|_{L^\infty} \\
&\times\{|\left[1-E_{+i\corr{\kappa},1/2}(\zeta)-E_{-i\corr{\kappa},1/2}(\zeta)\right]\mathfrak s_c(\zeta)| _{L^2(d\xi d\eta)}\left|\frac 1{p_\lambda}\right|_{L^2(\Omega^c_\lambda,d\xi d\eta)} \\
& +{|\left[1-E_{+i\corr{\kappa},1/2}(\zeta)-E_{-i\corr{\kappa},1/2}(\zeta)\right]\mathfrak s_c(\zeta) |}_{L^\infty(d\xi d\eta)}\left|\frac 1{p_\lambda}\right|_{L^1(\Omega_\lambda, d\xi d\eta)}\}. 

\end{array}
\eeq 
Therefore we obtain \eqref{E:infinity-lambda-i-spectal} and \eqref{E:spec-infinity}.
\end{proof}

\subsection{The Cauchy integral equation}

\begin{theorem}\label{T:cauchy-integral-eq}
If
\[
\ba{c}
u_0(x )=-2\corr{\kappa^2}\textrm{sech}^2\corr{\kappa}x,\, \kappa>0,\\
(1+|x|+|y|)^2\partial_y^j\partial_x^kv _0\in  {L^1\cap L^\infty},\ 0\le j,k\le 4 ,\\
|v_0 |_{{L^1\cap L^\infty}}\ll 1, \ v_0(x,y)\in\RR,
\ea
\] then the eigenfunction $\mathfrak m $  derived from Theorem \ref{T:normal-eigen}  satisfies $\mathfrak m (x,y,\lambda)\in W$ and the Cauchy integral equation
\beq\label{E:cauchy-integral-equation-sf}
\begin{array}{cl}
  \mathfrak{  m}(x,y,\lambda) =1+\frac{\mathfrak m_{ res }(x,y  )}{\lambda-\iota }  +\mathcal CT
 \mathfrak m , & \forall\lambda\ne  \iota,
\end{array}
\eeq   In particular, 
   the residue $\mathfrak m_{ res }(x,y  )$ at $\lambda=\iota$ and leading singularities $\mathfrak m_{\pm i\corr{\kappa},0}$ at $\pm i\kappa$ satisfy the constraints
\beq\label{E:+i-pogrebkov-x-function-new}
\ba{c}
 \frac{\mathfrak m_{ res }(x,y  )}{-i\corr{\kappa}}  =-1+\mathfrak m_{+i\corr{\kappa},0}(x,y,+i\corr{\kappa}+0^+e^{i\alpha})  - \mathcal C_{+i\corr{\kappa}+0^+e^{i\alpha}} T\mathfrak m  , \\
 \frac{\mathfrak m_{ res }(x,y  )}{-3i\corr{\kappa}}  =-1+\mathfrak m_{-i\corr{\kappa},0}(x,y,-i\corr{\kappa}+0^+e^{i\alpha}) - \mathcal C_{-i\corr{\kappa}+0^+e^{i\alpha}}T\mathfrak m ,\\
 \mathfrak m_{+i\corr{\kappa},0}(x,y,+i\corr{\kappa}+0^+e^{i\alpha})=  {\mathfrak s}_de^{-2\corr{\kappa}x}\mathfrak m_{-i\corr{\kappa},0}(x,y,-i\corr{\kappa}+0^+e^{i(\pi+\alpha)})
\ea
\eeq for $\forall -\frac\pi 2<\alpha <\frac {3\pi}2$, with  
 $W$,  $T$, $\mathfrak s_d$, and $\mathcal C$ defined by  Definition \ref{D:quadrature-hat} and \ref{D:spectral}.

\end{theorem}

\begin{proof}   
Theorem \ref{T:normal-eigen}  implies 
\begin{gather}
\ba{c}
\mathfrak m(x,y,\lambda)-\frac{\mathfrak m_{  res }(x,y  )}{\lambda-\iota} \quad\in L^\infty,\ea \label{E:regularize-m-ifty}\\
\ba{c}
  E_{0,n} T\mathfrak m(x,y,\lambda)\in  L^1(d\overline\lambda\wedge d\lambda), 
\ea\label{E:t-regularize-out-sing}
\end{gather} for $\forall   n>0$. Here $E_{z,a}$ is defined by Definition \ref{D:terminology}. 
Exploiting \eqref{E:t-regularize-out-sing} and applying \cite[\S I, Theorem 1.13, Theorem 1.14]{V62}, one 
derives
\beq\label{E:l-1-dbar}
\ba{c}
\partial_{\overline\lambda}\mathcal CE_{0,n} T\mathfrak m(x,y,\lambda) 
 =E_{0,n} T\mathfrak m(x,y,\lambda) \in  L^1(d\overline\lambda\wedge d\lambda).
 \ea
\eeq  Therefore, together with Theorem \ref{T:normal-eigen}, 
\beq\label{E:debar-initial}
\ba{c}
\partial_{\overline\lambda}\left[\mathfrak m(x,y,\lambda)-\frac{\mathfrak m_{ res }(x,y  )}{\lambda-\iota}-\mathcal CT\mathfrak m(x,y,\lambda)\right]=0.
\ea
\eeq
On the other hand, Lemma \ref{L:pogrebkov}, \ref{L:leading-iota}, and \ref{L:infinity-spectral} imply 
 \begin{gather}
 \ba{c}
| \mathcal CT\mathfrak m |\le C(1+|x|+|y|),\label{E:bdd-regularized}\ea\\
\ba{c}\mathcal C T\mathfrak m(x,y,\lambda)\to 0 \textit{ uniformly as $|\lambda|\to\infty$, $\lambda_R\ne 0$.}\label{E:bdd-regularized-inf}
 \ea
\end{gather}Applying \eqref{E:regularize-m-ifty}, \eqref{E:debar-initial}, \eqref{E:bdd-regularized},  and  Liouville's theorem,  one concludes  
\beq\label{E:unique-direct-problem-q}
\ba{c}
\mathfrak m(x,y,\lambda)=g(x,y)+\frac{\mathfrak m_{ res }(x,y)}{\lambda-\iota}  +\mathcal CT\mathfrak m(x,y,\lambda).
\ea
\eeq  
Equation \eqref{E:normal-ef-pde} and a direct computation  yield:
\beq\label{E:unique-computation}
\ba{rl}
&-u(x,y)\mathfrak m(x,y,\lambda)\\
=&\left(\partial_y-\partial_x^2+2i\lambda\partial_x\right)\mathfrak m(x,y,\lambda)\\
=&\left(\partial_y-\partial_x^2+2i\lambda\partial_x\right)
[g(x,y)+\frac    {\mathfrak m_{ res }(x,y)}{\lambda-\iota}    ]\\
&+\left(\partial_y-\partial_x^2+2i\lambda\partial_x\right)\mathcal CT\mathfrak m.
\ea
\eeq

Note that
\begin{eqnarray*}
&&\partial_x\mathcal CT\mathfrak m=\mathcal C[i(\lambda+\overline\lambda)T\mathfrak m+T(\partial_x \mathfrak m)],\\
&&\partial_x^2\mathcal CT\mathfrak m=\mathcal C[-(\lambda+\overline\lambda)^2T\mathfrak m+2i(\lambda+\overline\lambda)T(\partial_x \mathfrak m)+T(\partial_x ^2\mathfrak m)],\\
&&\partial_y\mathcal CT\mathfrak m=\mathcal C[ (\lambda^2-\overline\lambda^2) T\mathfrak m +T(\partial_y\mathfrak m)].
\end{eqnarray*}
 
Applying the Fourier transform theory, namely, \eqref{E:l2-linfty} and \eqref{E:high-der-lambda},  if $v(x,y)$ has $4$ derivatives in  $L^1\cap L^\infty$, then  
 \[
 \ba{l}
 (1-E_{+i\corr{\kappa},1/2}(\lambda)-E_{-i\corr{\kappa},1/2}(\lambda))(\lambda+\overline\lambda)\mathfrak s_c(\lambda),\\
 (1-E_{+i\corr{\kappa},1/2}(\lambda)-E_{-i\corr{\kappa},1/2}(\lambda))(\lambda+\overline\lambda)^2 \mathfrak s_c(\lambda),\\
 (1-E_{+i\corr{\kappa},1/2}(\lambda)-E_{-i\corr{\kappa},1/2}(\lambda))(\lambda^2-\overline\lambda^2)  \mathfrak s_c(\lambda),
 \ea
 \] 
are all bounded in $L^\infty\cap L^2(|\lambda_R|d\overline\lambda\wedge d\lambda)$. Therefore by \eqref{E:wick-infty-spec}, if  $\partial_y^j\partial_x^kv \in  {L^1\cap L^\infty}$, $ 0\le j,\ k\le 4 $,  one can adapt the proof for \eqref{E:bdd-regularized-inf} and  derive, as $|\lambda|\to\infty$, $\lambda_R\ne 0$, 
\[ 
\ba{c}
\left(\partial_y-\partial_x^2+2i\lambda\partial_x\right)\mathcal CT\mathfrak m \to 0.
\ea
\] 
So 
comparing growth in \eqref{E:unique-computation}, we conclude
\[\partial_x g(x,y)=0 
\]which turns \eqref{E:unique-direct-problem-q} into
\beq\label{E:lambda-infty}
\ba{c}
\mathfrak m(x,y,\lambda)-1=g(y)-1+\frac    {\mathfrak m_{ res }(x,y)}{\lambda-\iota}+ \mathcal C T \mathfrak m(x,y,\lambda).
\ea
\eeq
Fix $y$, and let $\epsilon>0$ be given. Let $\lambda\gg 1$, $\lambda_R\ne 0$, be chosen such that
\[
|\frac    {\mathfrak  m_{ res }(x,y)}{\lambda-\iota} +\mathcal C T \mathfrak m(x,y,\lambda)|<\frac \epsilon 2 
\]by \eqref{E:bdd-regularized-inf}. For this $\lambda$, by taking $x\to\infty$, and using the boundary property   \eqref{E:kp-line-normal-bdry-1}, we justify  $g\equiv 1$ and establish \eqref{E:cauchy-integral-equation-sf}.


\end{proof}

One direct corollary from  \eqref{E:cauchy-integral-equation-sf} is the uniform estimate  
\beq\label{E:corollary-cauchy-integral-equation}
\begin{array}{cl}
|\mathcal CT
 \mathfrak m |_{L^\infty}\le C.
\end{array}
\eeq

\begin{example}\label{Ex:cauchy-eq-v-0}
If $v_0(x,y)\equiv 0$, then $\gamma_+\equiv 0$, $\mathfrak s_c\equiv 0$. So \eqref{E:cauchy-integral-equation-sf} and \eqref{E:+i-pogrebkov-x-function-new} reduce to
\begin{gather}
\ba{c}
\mathfrak{  m}(x,y,\lambda) =1+\frac{\mathfrak m_{ res }(x,y  )}{\lambda-\iota } ,\ea\label{E:1}\\
\ba {c}\frac{\mathfrak m_{ res }(x,y  )}{-i\corr{\kappa}}  =-1+\mathfrak m_{+i\corr{\kappa},0}(x,y)    ,\ea\label{E:2} \\
 \ba{c}\frac{\mathfrak m_{ res }(x,y  )}{-3i\corr{\kappa}} =-1+\mathfrak m_{-i\corr{\kappa},0}(x,y) ,\ea\label{E:3}\\
 \ba{c}\mathfrak m_{+i\corr{\kappa},0}(x,y)=  {\mathfrak s}_de^{-2\corr{\kappa}x}\mathfrak m_{-i\corr{\kappa},0}(x,y) 
\ea\label{E:4}
\end{gather}  which yield
\beq\label{E:res-reg}
\ba{c}
 \mathfrak m_{-i\corr{\kappa},0}(x,y)=\frac{\frac 23}{1+e^{-2\corr{\kappa}x}},\ 
 \mathfrak m_{+i\corr{\kappa},0}(x,y)=\frac{-2}{1+e^{2\corr{\kappa}x}}, \\  \mathfrak m(x,y,\lambda)=
1+\frac{3i\kappa}{\lambda-2i\kappa}(1-\frac{\frac 23}{1+e^{-2\kappa x}}).
\ea
\eeq 
\end{example}


\end{document}